\documentclass[letter,11pt]{article}
\usepackage[bottom=1in,top=1in,left=1in,right=1in]{geometry}

\usepackage{times}

\usepackage{amsmath,amsthm,amssymb,soul}
\usepackage{mathtools}

\usepackage{hyperref}
\usepackage{xspace}
\usepackage{thm-restate}
\usepackage{subcaption}
\usepackage{enumerate} 
\usepackage[utf8]{inputenc}
\usepackage[T1]{fontenc}
\usepackage{paralist}

\newtheorem{theorem}{Theorem}
\newtheorem{definition}[theorem]{Definition}
\newtheorem{lemma}[theorem]{Lemma}

\newtheorem{observation}[theorem]{Observation}
\newtheorem{corollary}[theorem]{Corollary}

\usepackage{tikz}
\usetikzlibrary{calc}

\colorlet{shade}{gray!5}
\newcommand{\algorithmus}[2]{
{\centering
\vspace{0.1cm}
\begin{tikzpicture}
\node [fill=shade,draw]
{\parbox{\textwidth-0.5cm}{
{\bf #1}\vspace{-0.5cm}\\
#2}
};
\end{tikzpicture}\\
\vspace{0.2cm}
}
}

\newcommand*\samethanks[1][\value{footnote}]{\footnotemark[#1]}

\def\RR{{\mathbb R}}

\def\ZZ{{\mathbb Z}}

\def\MMM{{\mathcal M}}
\def\Pr{{\mathrm Pr}}
\def \d{{\mathrm d}}

\def\XX{{\mathbb X}}
\def\DD{{\mathcal D}}

\DeclareMathOperator{\loglog}{\log\hspace{-0.2em\log}}

\newcommand{\eps}{\epsilon}
\newcommand{\df}{\d_{dF}}
\newcommand{\dm}{\d_{\MMM}}
\newcommand{\Rdk}{\XX_k^d}

\usepackage{todonotes}

\title{Sublinear data structures for short Fr\'echet queries}

\author{Anne Driemel\thanks{Hausdorff Center for Mathematics, Bonn, Germany \url{driemel@cs.uni-bonn.de}} \qquad Ioannis Psarros\thanks{University of Bonn, Germany, \url{ipsarros@uni-bonn.de}, \url{melanieschmidt@uni-bonn.de}} \qquad Melanie Schmidt\samethanks[2]}

\begin{document}
\maketitle

\begin{abstract}
We study metric data structures for curves in doubling spaces, such as trajectories of moving objects
in Euclidean $\mathbb{R}^d$, where the distance between two curves is measured using the discrete Fr\'echet distance. 
We design data structures in an \emph{asymmetric} setting where the input is a curve (or a set of $n$ curves) each of complexity $m$ and the queries are with curves of complexity $k\ll m$. We show that there exist approximate data structures that are independent of the input size $N = d \cdot n \cdot m$ and we study how to maintain them dynamically if the input is given in the stream.

Concretely, we study two types of data structures: (i) distance oracles, where the task is to store a compressed version of the input curve, which can be used to answer queries for the distance of a query curve to the input curve, and (ii) nearest-neighbor data structures, where the task is to preprocess a set of input curves to answer queries for the input curve closest to the query curve. In both cases we are interested in approximation. 
For curves embedded in Euclidean $\mathbb{R}^d$ with constant $d$, our distance oracle uses space in $\mathcal{O}((k \log(\epsilon^{-1}) \epsilon^{-d})^k)$ ($\epsilon$ is the precision parameter). The oracle performs $(1+\epsilon)$-approximate queries in time in $\mathcal{O}(k^2)$ and is deterministic. We show how to maintain this distance oracle in the stream using polylogarithmic additional memory. In the stream, we can dynamically answer distance queries to the portion of the stream seen so far in  $\mathcal{O}(k^4 \log^2 m)$ time. 
We apply our techniques to the second problem, approximate near neighbor (ANN) data structures, and  achieve an  exponential improvement in the dependency on the complexity of the input curves compared to the state of the art. For polygonal curves in Euclidean $\mathbb{R}^d$, we present a data structure with space in $n\cdot  \mathcal{O}\left({k{d}^{3/2}}{\epsilon^{-1}}\right)^{dk}$ plus additional space in $\mathcal{O}(dnm)$ for storing the input,  and query time in $\mathcal{O}\left(dk\right)$. Our ANN data structure is randomized, but in the Euclidean case we can derandomize it at little extra cost.
    
\end{abstract}

\thispagestyle{empty}
\newpage
\setcounter{page}{1}
\section{Introduction}

Given an input, it is often desired to store it in a data structure of small size with the ability to retrieve some meaningful information in an online fashion. In this paper, we consider inputs that consist of one or more trajectories, i.e., observed paths of moving entities.
The last five years have seen a raised interest in data structures for trajectory processing under the Fr\'echet distance \cite{deBerg-2013fast, Gudmundsson-fqt-15, giscup2017, BaldusB2017, DutschV17, BuchinDDM17, AD18, DPP-VC-19, astefanoaei2018multi,Indyk-approxnn-02,Driemel-lshc-17,emiris2018, FFK19, driemel-jaywalking}. Conceptually, these data structures define a metric space where all reparametrizations of a curve are equivalent and the Fr\'echet metric measures the distance between any pair of equivalence classes. This invariance under reparametrizations makes the Fr\'echet distance especially suitable for comparing trajectories with varying speeds, i.e., for tracking data of moving animals or sports players. However, it also makes it challenging to design efficient data structures, since many standard techniques for vector spaces, such as recursive partitioning techniques, cannot be easily applied. 
In the relevant applications, the query curves tend to be of low complexity. For example, in sports analysis, one may wish to query with a straight path from A to B on the playing field~\cite{Berg-spq-15}. It is tempting to believe that under such an assumption on the complexity of queries it should be possible to build data structures with proportionally small space and low query time. This hypothesis is the starting point of our study.
In the above mentioned literature there are three types of queries studied: (i) \emph{range queries with metric balls} \cite{deBerg-2013fast, Gudmundsson-fqt-15, giscup2017, BaldusB2017, DutschV17, BuchinDDM17, AD18, DPP-VC-19} (ii) \emph{nearest-neighbor queries} \cite{astefanoaei2018multi,Indyk-approxnn-02,Driemel-lshc-17,emiris2018, FFK19}, and (iii) \emph{distance queries}~\cite{driemel-jaywalking, astefanoaei2018multi}.   
In this paper, we design efficient approximate data structures for queries of type (ii) and (iii) under the assumption that queries have low complexity (measured in the length of the point sequence). For this, we focus on the discrete Fr\'echet distance, a simplified variant of the distance metric which is defined on point sequences. 
Data structures of type (iii) are sometimes called \emph{distance oracles} and are related to the notions of \emph{sketching} and \emph{coresets}: assuming a certain metric and an input point, the aim is to compute a compressed representation which allows for estimating the distance from a new query point.
Bravermann et al.~\cite{braverman2019oneway} recently showed that for computing an $\alpha$-approximation to the related Dynamic time warping distance in the randomized one-way communication protocol, a sketch requires $\Omega(n/\alpha)$ bits---assuming the length of the query curve can be as large as $n$. 
%
Our results show that studying these problems in the asymmetric setting (by separating the two parameters that define the length of the query and the length of the input curves) can lead to data structures that are much more efficient when we assume small query complexity. 
We show that in the static setting, where the input is known in advance, we can build a data structure with size \emph{independent} of the length of the input curves.
Secondly, we study how to maintain a small data structure---or sketch---in a stream, that allows us to answer queries for the discrete Fr\'echet distance from a short query curve to the part of the stream which has been discovered so far. We show that we can maintain a sketch that only depends \emph{poly-logarithmically} on the portion of the input seen so far and can be used to compute a $(1+\eps)$-approximation to the discrete Fr\'echet distance for any query curve.  
Our streaming algorithm uses the \emph{merge-and-reduce} framework that has found applications in many streaming algorithms, for example in the area of clustering. It goes back to the work of Bentley and Saxe~\cite{BS80}. In clustering, merge-and-reduce has been successfully applied for problems that allow the construction of \emph{composable coresets}: A coreset is a small summary of a data set that behaves similarly to the original set with respect to a predefined cost function. A coreset construction is \emph{composable} if two instances of the coreset can be merged without inducing additional error. Our sublinear data structure for distance queries can be interpreted as a composable coreset. The first work to use merge-and-reduce for composable coresets in clustering is~\cite{HPM04}, and our application is similar to theirs (except that how we achieve composability is much more complex). 
The second problem we study is building approximate data structures of type (ii) for the discrete Fr\'echet distance.
In particular, we study the $(1+\eps)$-approximate near neighbor, also referred to as ANN, is an element of the input of which the distance to the query is most a factor $(1+\eps)$ further away than a fixed parameter $r$. The approximate nearest neighbor problem is known \cite{HIM12} to reduce to a sequence of approximate near neighbor problems. This problem is fundamental in applications such as machine learning, information retrieval and classification. 
Our data structure is especially useful if the input curves are very long compared to the query curves. Intuitively, this corresponds to the input living in high dimensional space.\footnote{This can be formalized in terms of the doubling dimension (see also Section~\ref{sec:doublingspaces} in the Appendix).} The standard (and baseline) approach to nearest neighbor searching in high-dimensional space is a linear scan over the data during query time. The aim is therefore to achieve sublinear \emph{query complexity}, see also the discussion in~\cite{Shakhnarovich:2006:NML:1197919}. In fact, we will show query complexities \emph{independent} of the length of the input curves $m$ and logarithmic in the number of input curves $n$.
\subsection{Preliminaries}\label{sec:prelim}
Our results distinguish between inputs in Euclidean metrics in constant dimension and the case of metrics with bounded doubling dimension given in the weakly explicit model defined below.
In the Euclidean case, we denote a ball of radius $r$ around $x \in \mathbb{R}^d$ by $b(x,r)$. We use $\XX_m^d =\left(\RR^d \right)^m $ and treat the elements of this set as ordered sets of points in $\RR^d$ of size $m$ called \emph{curves}. 

In the metric case, we assume a metric space $(\MMM^m,\dm)$, write a curve $p$ with $m$ vertices as $p=p_1,\ldots,p_m$ and denote the space of all curves by $\MMM^m$. Furthermore we denote a ball of radius $r$ around $x \in \MMM$ by $b_{\MMM}(x,r)$. We define doubling spaces as follows.
\begin{definition}[Doubling constant]
Consider any metric space with ground set $X$.
The {\em doubling constant} of $X$, denoted $\lambda_X$, is the smallest integer $\lambda_X$ such that for any
$p \in X$ and $r > 0$, the ball $b_{\MMM}(p, r)$ (in $X$) can be covered by at most $\lambda_X$ balls of radius $r/2$ centered at points in $X$. The {\em doubling dimension} of $X$ is $d_X = \log \lambda_X$. $(\MMM,\d_{\MMM})$ is a {\em doubling space} if $\lambda_X$ is a constant.
\end{definition}
We assume the existence of a constant-time oracle that gives us access to the metric space. We refer to the two computational models relevant for our work as follows:
\begin{inparaenum}[(a)]
    \item {\em black-box model} (\cite{CG06,HM06,KL04}): there exists a constant-time {\em distance oracle} for the metric space that reports the pairwise distance for any two points,
    \item {\em weakly explicit model} (\cite{AMVX08}): there exists a distance oracle and a {\em doubling oracle} for the metric space. Given any ball in
the metric space $\MMM$, the doubling oracle returns in time $\lambda_{\MMM}$ a covering with $\lambda_{\MMM}$ balls of half the radius.  
\end{inparaenum}

Note that the class of doubling spaces includes the Euclidean metric space. For the sake of generality, we consider curves embedded in a doubling space. We make a distinction in some cases where the Euclidean metric allows for better results. Now we define the discrete Fr\'echet distance. 

 \begin{definition}[Traversal]
 Given $p=p_1, \ldots, p_{m}\in \MMM^m$ and $q=q_1, \ldots, q_{k}\in \MMM^k$, a traversal $T=(i_1,j_1),\ldots,(i_t,j_t)$ of $p$ and $q$ is a sequence of pairs of indices referring to a pairing of points from the two sequences such that:
\begin{compactenum}[(i)]
 \item $i_1,j_1=1$, $i_t=m$, $j_t=k$. 
 \item $\forall (i_u, j_u)\in T:$ $i_{u+1}-i_u \in \{0,1\}$ and $j_{u+1}-j_u \in \{0,1\}$.
 \item $\forall (i_u, j_u)\in T:$ $(i_{u+1}-i_u)+(j_{u+1}-j_u)\geq1$.
\end{compactenum} 
For a traversal $T$, we call $\max_{(i_u,j_u)\in T} \d_{\MMM}(p_{i_u},q_{j_u})$ the \emph{cost} of $T$.
 \end{definition} 

\begin{definition}[Discrete Fr\'{e}chet distance]\label{Ddist}
 Let $(\MMM, \d_{\MMM})$ be a metric space. 
 Given $p \in \MMM^m$ and $q\in \MMM^k$, we define the Discrete Fr\'{e}chet Distance between $p$ and $q$ as follows:
 \[
 \df(p,q)= \min_{T\in\mathcal{T}}  \max_{(i_u,j_u)\in T} \dm(p_{i_u},q_{j_u})  ,
 \]
 where $\mathcal{T}$ denotes the set of all possible traversals for $p$ and $q$. Thus, $\df(p,q)$ is the minimum cost of any traversal of $p$ and $q$.
 \end{definition}

Notice that the discrete Fr\'{e}chet distance defines a pseudo-metric: the triangular inequality is satisfied, but distinct curves may have zero distance. For our purposes, it is sufficient to consider the metric space which is naturally induced by this pseudo-metric: two polygonal curves are equivalent if their discrete Fr\'{e}chet distance is zero. 
Furthermore notice that adding copies of a vertex to a curve $X$ does not change the discrete Fr\'echet distance of $X$ to any other curve. Thus, if we have a distance oracle for all query curves in $\MMM^k$, it also works for all queries of size $k' \le k$.

It is well known~\cite{eiter1994computing} that computing the discrete Fr\'echet distance between two point sequences in 
a metric space
can be done using dynamic programming in $\mathcal{O}(k m T_{dist})$ time, where $m$ and $k$ respectively denote the length of the two sequences and $T_{dist}$ denotes the time it takes to compute a distance in the underlying metric space. If this is the Euclidean metric, then $T_{dist}=\mathcal{O}(d)$ and the above is known to be optimal up to lower order factors unless the Strong Exponential Time Hypothesis (SETH) fails~\cite{bringmann2014}.

Our data structures make use of the concept of simplifications which we define as follows.

\begin{definition}[$\alpha$-approximate $k$-simplification]
Given a polygonal curve $p\in \MMM^m$, a polygonal curve $p' \in \MMM^k$ is an $\alpha$-approximate  $k$-simplification, if for any $q \in \MMM^k$ it holds that
\[ \d_{dF}(p,p') ~\leq~ \alpha \cdot d_{dF}(p,q)\]
\end{definition}

Algorithms for curve simplifications are a well-studied topic reaching back to the work of Imai and Iri, see~\cite{HM88} and references therein.  It is known how to compute optimal simplifications under the (discrete) Fr\'echet distance, 
see~\cite{altgodau-95,BeregSimp,Godau91, kreveld2018simp}. 
However, we are interested in a streaming algorithm for our purposes. 
Abam et al.~\cite{Abam2010} show how to maintain approximate simplifications for different error measures including the \emph{continuous} Fr\'echet distance in the streaming setting. Their algorithm maintains a $2k$-simplification that is $(4\sqrt{2}+\epsilon)$-approximate compared to an optimal $k$-simplification, using $\mathcal{O}(k^2 \epsilon^{-0.5} \log^2 \epsilon^{-1})$ space and $\mathcal{O}(k \epsilon^{-0.5} \log^2 \epsilon^{-1})$ amortized update time per point.
We observe that for the \emph{discrete} Fr\'echet distance, we can obtain a simple algorithm that uses very little space and update time. 
We describe the algorithm in Section~\ref{sec:simplification}. The following theorem states the result.

\begin{restatable}{theorem}{thmstreamingsimplification}\label{thm:streamingeightapprox}
Let $X\in \MMM^m$ be given as a (possibly infinite) stream in the black-box distance oracle model for a general metric space. 
There is a streaming (one-pass) algorithm that computes an $8$-approximate $k$-simplification of $X$. The algorithm needs $\mathcal{O}(k)$ memory, and updates the curve simplification after reading a point in time $\mathcal{O}(k)$. Processing a curve $X\in \MMM^m$ thus has a total running time of $\mathcal{O}(mk)$.
\end{restatable}

\subsection{Previous work}\label{sec:previouswork}

It is beyond the scope of this paper to give a general overview of coresets and sketching techniques. We refer the interested reader to the surveys by Phillips~\cite{phillips2016coresets},  by Munteanu and Schwiegelshohn~\cite{Munteanu2018}, and  by Muthukrishnan~\cite{Muthu05}.
We summarize what is known on distance oracles for the Fr\'echet distance in Table~\ref{tab:distoracles}.
A $(1+\eps)$-approximate distance oracle was presented by Driemel and Har-Peled~\cite{driemel-jaywalking}. This data structure also supports queries to subcurves of the input curve, specified by including start and end points in the query. However, the space usage is super-linear in the input size and queries curves are restricted to be line segments. 
As we will see, a simple constant-factor-approximate distance oracle of size $O(k)$ can be obtained by storing a simplification of the input curve and by using this curve as a proxy. 
We cite the algorithms of Bereg et al.~\cite{BeregSimp} and Abam et al.~\cite{Abam2010} for computing simplifications. For a more detailed discussion of this approach refer to Section~\ref{sec:our_results_distoracle}.
Note that~\cite{Abam2010} and \cite{driemel-jaywalking} use the continuous Fr\'echet distance. However, we still include them for the purpose of comparison.

\begin{table}[tb]\centering\small
\def\arraystretch{1.2}
\caption{\label{tab:overview}Summary of previous data structures for Fr\'echet queries compared to our results. In the bounds below, $n,m,k$ and $d$ are defined as follows: the input is a set of curves $P \subset \XX^d_m$ with $|P|=n$; the query curves are in $\XX^d_k$, where $\XX^d_m = (\RR^d)^m$ is the set of point sequences in $\RR^d$ of length $m$. 
}
\begin{subtable}{\linewidth}\centering
{\begin{tabular}{|c|c|c|l|} \hline
   Space & Query & Approx. & Comments \\
   \hline\hline
   $\mathcal{O}(k)$ & $\mathcal{O}(k^2)$ & $3$ & static, using\,\cite{BeregSimp} \\
   $\mathcal{O}(k^2)$ & $\mathcal{O}(k^2)$ & $\mathcal{O}(1)$ & stream., using\,\cite{Abam2010} \\
    $\mathcal{O}((\eps^{-d} \log \frac{1}{\eps})^2 m)$ &  $\mathcal{O}(\eps^{-2}\log m \loglog m )$ & $1+\eps$ & static, $k=2$, \cite{driemel-jaywalking}\\
   \hline \hline
  $ \mathcal{O}( (k \log(\frac{1}{\eps}) \eps^{-d} )^k ) $ &   $\mathcal{O}( k^2 + \log \epsilon^{-1})$ & $1+\eps$ & static, Thm.\,\ref{thm:distoracle}
  \\
          $ \mathcal{O}( (k \log(\frac{1}{\eps}) \eps^{-d} )^k ) $ & $\mathcal{O}(k^2 + \log \epsilon^{-1})$ & $1+\eps$ & one-pass, Thm.\,\ref{thm:onepassDistOracle} \\
   $ {\mathcal{O}}\left(\log^2 m \cdot  k^k \cdot (\frac{\log m}{\eps})^{dk} \cdot (\log (\frac{\log m}{\eps}))^k  \right) $ 
    & $\mathcal{O}(k^4 \log^2 (m/\epsilon))$ & $1+\eps$ & stream., Thm.\,\ref{thm:streamingDistOracle} \\     

   \hline
   \end{tabular}}
\caption{\label{tab:distoracles} Overview of results for distance oracles. All of the results are deterministic. The data structure in~\cite{driemel-jaywalking} supports distance estimation to subcurves, while the other results are for distance estimation to the entire curve. The query times are stated under the assumption that we compute $\lceil x\rceil$ and $\lfloor x\rfloor$ in $\mathcal{O}(1)$ time.}
\end{subtable}
\begin{subtable}{\linewidth}\centering
{\begin{tabular}{|c|c|c|l|} \hline
   Space & Query & Approx. & Comments \\ 
   \hline\hline 
   ${\mathcal{O}}( m^2 |U| )^{m^{1-o(1)}}\cdot \mathcal{O}( n^{2-o(1)})$    &$\left(m \log n\right)^{\mathcal{O}(1)}$ &  $\mathcal{O}(1)$ & metric, det., \cite{Indyk-approxnn-02}
   \\
$\mathcal{O}(n \log n + dnm)$ & $\mathcal{O}(dk \log n)$  & $\mathcal{O}(k)$ & $\ell_2^d$, rand., \cite{Driemel-lshc-17}
\\ 
	 ${\mathcal{O}}(2^{4md} n )$ &${\mathcal{O}}(2^{4md} \log n )$  &  $\mathcal{O}(d^{3/2})$  & $\ell_2^d$, rand., \cite{Driemel-lshc-17}
    \\ 
   $\tilde{\mathcal{O}}( n )\cdot \left( \frac{d}{\log m} +2\right)^{\mathcal{O}\left(m^{\left(1+\frac{1}{\eps}\right)} d \log\left(\frac{1}{\eps}\right)\right)}$ &  $\tilde{\mathcal{O}} \left(dm^{(1+\frac{1}{\eps})} 2^{4m} \log n \right) $ & $1+\eps$ & $\ell_2^d$, rand., \cite{emiris2018} 
    \\ 
    $n \cdot  \mathcal{O}\left(\frac{1}{\eps}\right)^{dm}$ & $\mathcal{O}(md \log (nm/\eps))$ & $1+\eps$ & $\ell_2^d$, det., \cite{FFK19} 
    \\\hline\hline
   $n\cdot  \mathcal{O}\left(\frac{k{d}^{3/2}}{\eps}\right)^{dk} +\mathcal{O}(dnm)$ &  $\mathcal{O}\left(dk\right)$ & $1+\eps$ & $\ell_2^d$, rand., Thm.\,\ref{TannDFDhd} 
   \\
    $d^{3/2}n k \eps^{-1} \cdot \mathcal{O}\left(\frac{k{d}^{3/2}}{\eps}\right)^{dk}+\mathcal{O}(dnm)$ &  $\mathcal{O}\left(d^{5/2} k^2 \eps^{-1} (\log n + kd \log \frac{kd}{\eps}) \right) $ & $1+\eps$ & $\ell_2^d$, det., Thm.\,\ref{TdeterministicannDFDhd}
\\\hline
\end{tabular}}
\caption{\label{tab:compar} Overview of results for ANN. The result by Indyk~\cite{Indyk-approxnn-02} is tuned to optimize the approximation factor. $U$ denotes the ground set of the metric, where the Fr\'echet distance is derived from a finite metric space $(U,d)$.  
}
\end{subtable}
\end{table}

Previous results on data structures for ANN search under the discrete Fr\'{e}chet distance, are summarized in Table \ref{tab:compar}.  The first result from 2002 by Indyk \cite{Indyk-approxnn-02} achieved approximation factor 
$\mathcal{O}((\log m + \log \log n)^{t-1})$, where $m$ is the maximum length of a curve, and $t>1$ is a trade-off parameter. 
Table~\ref{tab:compar} states these bounds for appropriate $t=1+o(1)$, hence a constant approximation factor.  
More recently, in 2017, Driemel and Silvestri~\cite{Driemel-lshc-17} showed a locality-sensitive hash family for the discrete Fr\'echet distance in the standard setting where the underlying metric is Euclidean. Their basic approach achieves approximation factor $\mathcal{O}(k)$, where $k$ is the length of the query curve. They show how to improve the approximation factor to $\mathcal{O}(d^{3/2})$ at the expense of additional space usage. Table \ref{tab:compar} states the two extremes of this tradeoff. 
A follow-up result by Emiris and Psarros~\cite{emiris2018} achieves a ($1+\epsilon$) approximation, at the expense of increasing space usage. This approach is based on weak randomized embeddings from $\ell_2$ to other $\ell_p$-metrics. 
Most recently, Filtser et al.~showed how to build a $(1+\eps)$-approximate deterministic data structure for the problem using space in $ n \mathcal{O}(1/\eps)^{md}$ and query time in $\mathcal{O}(m d \log(nm/\eps))$. Note that all previous results have some exponential dependency on $m$, the length of the input sequences. The improvement is even more dramatic in the query times, since our query times are independent of $m$, while previous $(1+\eps)$-approximate data structures require query time linear in $m$, or even exponential in $m$.
There has been extensive work on ANN data structures for doubling spaces, which are general metric spaces with bounded dimension. 
In the literature it is common to assume a constant-time distance oracle (black-box model), as done by Krauthgamer and Lee~\cite{KL04}, Har-Peled and Mendel~\cite{HM06}, and Cole and Gottlieb~\cite{CG06}. The slightly more restricted weakly explicit model was proposed by Arya et al.~\cite{AMVX08}. 
In order to put our results into this larger context, we give bounds on the doubling dimension of the discrete Fr\'echet distance and analyze which bounds can be obtained following this approach. We include the details in the Appendix (Section~\ref{sec:doublingspaces}).
We use some of these data structures as building blocks for the extension of our data structures to the discrete Fr\'echet distance derived from doubling metrics.

\subsection{Our results}

In the following, we discuss our main results on data structures for queries under the discrete Fr\'echet distance in the short queries regime.  Table~\ref{tab:overview} gives an overview of the obtained bounds compared to the state of the art. For ease of comparison, the table only shows the results for the common case that the discrete Fr\'echet metric is derived from the Euclidean metric.

\subsubsection{Distance Oracles} \label{sec:our_results_distoracle}

We show how to construct $(1+\eps)$-approximate distance oracles using sublinear space. In the static case the final data structure is even \emph{independent} of the input size. This leads to the following theorem. The proof of the theorem can be found in Section~\ref{sec:distance:oracle}.

\begin{restatable}{theorem}{thmTimeSpace}\label{thm:distoracle}
Let $(\MMM,\dm)$ be a metric space with constant doubling dimension $d \in \mathbb{N}$ in the weakly explicit model, let $\alpha \ge 8$ be a constant.
For any curve $X \in \MMM^m$ and any $k \in \mathbb{N}^{\ge 1}$, we can compute a data structure
of size $\mathcal{O}(f_s(k,\eps,\alpha))$,
    in preprocessing time $\mathcal{O}(mk \cdot f_s(k,\epsilon,\alpha))$ and 
    space $\mathcal{O}(mk+f_s(k,\epsilon,\alpha))$, and 
    with query time $\mathcal{O}(f_q(k))$
which for any query curve $q \in \MMM^{k'}$ with $k'\le k$ returns a value $d'(q)$ such that 
\[\df(q,x) \le d'(q) \le (1+\epsilon)\cdot \df(q,X).\]
In the Euclidean case, $f_q(k,\eps,\alpha):= k^2+\log\frac{\alpha}{\epsilon}$ and $f_s(k,\eps,\alpha):= k^k \cdot (\log \frac{\alpha}{\epsilon})^k \cdot (\frac{\alpha}{\epsilon})^{d k}$. 
In the metric case, $f_q(k,\eps,\alpha)= \mathcal{O}(k^2 + k \cdot \log (k \cdot \frac{\alpha}{(2/3)\epsilon}))$ and $f_s(k,\eps,\alpha):= k^k \cdot (\log \frac{\alpha}{(2/3)\epsilon})^k \cdot (\frac{\alpha}{(2/3)\epsilon})^{d_{\MMM} k}$.
\end{restatable}

In Section~\ref{sec:streaming} we show that it is possible to maintain such distance oracles in the stream. To the best of our knowledge ours constitutes the first streaming algorithm for the Fr\'echet distance that can perform $(1+\eps)$-approximate distance queries while using sublinear additional memory. 
In particular, for the streaming algorithm, we consider two scenarios: In the classical setting, we do not know the length of the stream in advance and we want to be able to answer dynamic distance queries in the stream. In this case, our algorithm balances the memory requirements of the individual distance oracles with the time required to perform dynamic distance queries on the concatenation of these distance oracles. This leads to the following theorem. The proof of the theorem can be found in Section~\ref{sec:streaming:implications}.

\begin{restatable}{theorem}{thmDistOracleStreamingUnKnownM}
\label{thm:streamingDistOracle}
There is a streaming algorithm that computes a distance oracle for a curve $X\in \XX^d_m$ given as a stream with $m=|X|$, where $m$ is not known in advance, with the following properties. The algorithm uses additional memory in 
\[
\mathcal{O}\Big(\log^2 m \cdot k^k \cdot \big(\log \big(\frac{\log m}{\eps}\big)\big)^k  \cdot (\frac{\log m}{\eps})^{dk}\Big)
\]
 and the total time used by the algorithm to compute a distance oracle in the stream is bounded by 
$
\mathcal{O}( m  \cdot k^{k+3}  \cdot (\frac{\log m}{\eps})^{dk + 1} \cdot (\log (\frac{\log m}{\eps}))^k)
$. 
Using the data structures maintained in the stream, one can answer distance queries at any point in the stream. For any query curve $q \in \XX^d_k$, the algorithm outputs a value $d'(q)$ in query time $\mathcal{O}( k^4 \cdot \log^2 m/\epsilon )$ and additional memory in $\mathcal{O}(k \cdot \log^2 m + k^2)$ such that 
\[\df(q,x) \le d'(q) \le (1+\epsilon)\cdot \df(q,X).\]
\end{restatable}

Secondly, in the \emph{one-pass scenario} we are interested in computing a static distance oracle using one pass over the input and while using sublinear additional memory. In this case, we aim to produce a small data structure with little dependency on the length of the input (at the expense of a higher dynamic query time in the stream and higher additional memory in the stream). The proof of the following theorem can be found in Section~\ref{sec:streaming:implications}.

\begin{restatable}{theorem}{thmOnePassDistOracle}\label{thm:onepassDistOracle}
There is a one-pass algorithm that computes a distance oracle for a curve $X\in\XX^d_m$ given as a stream of length $m$. For $1 < s < \log_2 m$, the algorithm uses additional memory in 
$\mathcal{O}( m^{1/s} \cdot k^k (\log \eps^{-1})^k  \eps^{-dk})$ 
and has running time in 
\[ 
\mathcal{O}\left( 
m^{(1+\frac{1}{s})} \cdot k^{(k+3)} \cdot (\log \eps^{-1})^k \cdot \eps^{-dk}
 \right)
\]
The resulting distance oracle has size 
$\mathcal{O}( k^k (\log \eps^{-1})^k  \eps^{-dk} )$. 
For any query curve $q \in \XX^d_k$ the data structure outputs a value $d'(q)$ in query time $\mathcal{O}(k^2+\log \epsilon^{-1})$ such that
\[\df(q, X) \le d'(q) \le (1+\eps) \df(q,X). \]
\end{restatable}

We remark that the general framework and the composition techniques for our streaming algorithm readily extend to other related distance measures, such as Dynamic time warping. However, we are not aware of any sublinear distance oracles that could be used inside the framework for such distance measures, even in the short queries regime.

\subsubsection{Techniques I} \label{sec:our_results_distoracle_techniques}
To establish the results stated in Section~\ref{sec:our_results_distoracle} above, we combine and extend multiple known techniques like exponential grids for points or the merge-and-reduce technique for obtaining streaming algorithms. 
First, we observe that a small $O(k)$-size distance oracle that answers queries up to a \emph{constant} approximation factor can be easily obtained from the following observation, which is implied by the triangle inequality.

\begin{observation}\label{simpleOracle}
Given $p \in \MMM^m$ and $q \in \MMM^k$, if $p'$ is an $\alpha$-approximate $k$-simplification of $p$, then (using the triangle inequality) it holds that 
\[ \d_{dF}(p,q) \leq \d_{dF}(p,p') + \d_{dF}(p',q) \leq 2\d_{dF}(p,p') + \d_{dF}(p,q)  \leq (2\alpha+1) \cdot \d_{dF}(p,q) \]
\end{observation}

The idea is to store a $k$-simplification $p'$ of the input curve $p$ together with the distance $d_F(p',p)$. By Observation~\ref{simpleOracle}, an approximation to the distance between a query curve $q$ and $p$ is now given by $d_F(p,p')+d_F(p',p)$. That is, we can use the simplification as a proxy curve. However it seems that, without any extra information on the input curve, this approach leads to constant factor approximations only (see also the discussion in Section~\ref{sec:flawed} in the Appendix).  
Building upon this idea, we use an exponential grid centered on each vertex of a compressed curve (simplification) and precompute all distances for all relevant curves on the grid. Using such distance information for a set of proxy curves we can then answer $(1+\eps)$-approximate distance queries for any query curve, at the cost of some space overhead for the look-up table. 
Next, we show that this data structure can be maintained in the stream.
Here, the main technical challenge that we solve is that, initially, our sublinear data structure does not seem \emph{composable}, which means that given two or more data structures for subsequent parts of a curve, it is not obvious how to build a data structure (or even how to answer queries) for the concatenation of the different parts. We show how to obtain composability by using dynamic programming on partitions of the query curve (Section~\ref{sec:composition}). Since this merging procedure is costly, we then observe that a non-binary computation tree can be beneficial to improve the performance of the merge-and-reduce algorithm. We use a trade-off parameter controlling the height of the computation tree to balance the size of the final distance oracle produced by the algorithm with the additional memory and query time used during the streaming which results from maintaining several distance oracles on different parts of the stream. 
Each time we rebuild our data structure within the merge-and-reduce framework we need to have a simplification of the input curve for the right portion of the input stream at our disposal. We show how to maintain such simplifications in the stream in Section~\ref{sec:simplification} and we carefully interleave the two streaming algorithms in Section~\ref{sec:streaming:algo}---leading to our main streaming algorithm.
We generalize our approach to curves defined on metric spaces of bounded doubling dimension by using the data structure of Arya et al.~\cite{AMVX08} for generating and querying the set of proxy curves. The generalized result is also stated in Theorem~\ref{thm:streamingDistOracle}.

\subsubsection{ANN data structures}

Refer to Table~\ref{tab:compar} for an overview of our results on nearest neighbor data structures in comparison to previous work. Note that all previous results achieving an approximation factor $(1+\eps)$ have some exponential dependence on the maximum complexity of the input curves, either in space or in query time. Our data structure has only linear dependency on $m$, the complexity of the input curves and only exponential dependency on $k$, the complexity of the query curves.
We give the exact theorem below. The proof of the theorem can be found in Section~\ref{subsection:ann:frechet}. 
Note that our ANN data structure is randomized: the performance guarantees correspond to constant probability of success which can be amplified by repetition.

\begin{restatable}{theorem}{thmANNdfdhd}
\label{TannDFDhd}
Given as input a set of $n$ polygonal curves $P\subset \XX_m^d$, and an approximation parameter $\eps>0$, 
there exists a randomized data structure with space  in $n\cdot \mathcal{O}\left(\frac{kd^{3/2}}{\eps}\right)^{kd}$ plus additional space in $\mathcal{O}(dnm)$ for storing the input, {preprocessing time} in  $dnmk\cdot \mathcal{O}\left(\frac{kd^{3/2}}{\eps}\right)^{kd}$, and 
query time in $\mathcal{O}(dk)$, for the ANN problem under the discrete Fr\'{e}chet distance. 
For any query curve $q\in\XX_k^d$, the preprocessing algorithm succeeds with constant probability.
\end{restatable}

We also show that our data structure can be derandomized at little extra cost. This leads to the following theorem. The proof can be found in Section~\ref{subsection:ann:frechet}.

\begin{restatable}{theorem}{thmANNddfdhd}
\label{TdeterministicannDFDhd}
Given as input a set of $n$ polygonal curves $P\subset \XX_m^d$, and an approximation parameter $0<\eps<1$, 
there exists a deterministic data structure with space  in $ \left(d^{3/2}  n k \eps^{-1} \right)\times \mathcal{O}\left(\frac{kd^{3/2}}{\eps}\right)^{kd}$ plus additional space in  $\mathcal{O}(dnm)$ for storing the input, preprocessing time in  $ \mathcal{O}\left( d^{5/2}nmk  \eps^{-1}\right)\times \mathcal{O}\left(\frac{kd^{3/2}}{\eps}\right)^{kd}$, and 
query time in $\mathcal{O}\left(\frac{k^2d^{5/2}}{\eps} (\log n + kd \log \frac{kd}{\eps}))\right)$, for the ANN problem under the discrete Fr\'{e}chet distance, for query curves in $\XX_k^d$. 
\end{restatable}

Finally, in Section~\ref{subsection:anndoubling} we further extend our techniques to handle curves in doubling spaces. The results are stated in Theorems \ref{TannDFDhdDDweak} and 
\ref{TannDFDhdDD} in the same section.

\subsubsection{Techniques II}
We describe the techniques used to establish our results on ANN data structures.
For curves in Euclidean space, we use the  $\mathcal{O}(k)$-approximate locality-sensitive hashing scheme proposed by Driemel and Silvestri~\cite{Driemel-lshc-17} to snap the input  to a (coarse) randomly shifted grid. After that, each bucket of the hash table is refined further using (finer) $\eps$-grids. We precompute and store approximate answers for all possible $k$-paths on the vertices of this grid. At query time, we aim to find the grid curve that lies closest to the query and return the precomputed answer.
We remark that Filtser et al.~\cite{FFK19} use a similar strategy. However, their data structure applies the $\eps$-grid to all points of the input sequences, directly. This results in an exponential dependency on $m$, while our data structure uses space linear in $m$ and query time independent of $m$. 
Indeed, our approach follows along the lines of the main idea described in Section~\ref{sec:our_results_distoracle_techniques} in that the coarse grid yields a compressed version of the input curves akin to $k$-simplifications. 
Next, we derandomize our data structure. This is achieved by observing that if we first discretize the ambient space, then the number of distinct outcomes produced by choosing a randomly shifted grid is bounded by a small number. This can be done while incurring little additive error. Concretely, we first snap points to an $\eps$-grid and then build the data structure for all possible shifts yielding distinct outcomes. It remains to observe that the perfect hashing used by our data structure can be replaced by plain binary search with a $O(\log n)$ overhead in the query time. 

Finally, in Section~\ref{subsection:anndoubling} we show how to extend our techniques for building ANN data structures to handle curves in general doubling metrics. This incurs a slight increase in query time since we cannot simply snap the query to a canonical grid. However, we can use standard techniques for doubling spaces.
The high-level idea of our solution is very similar to the one used above. We use nets, in order to discretize the input space, and a net-hierarchy which allows for a fast implementation of a $\Delta$-bounded-diameter random partition. Such partitions are quite common in the literature (see e.g.\ \cite{H11}, Chapter 26). They are particularly interesting because they satisfy the following two properties: 
near points belong to the same cluster with good probability, and each cluster has bounded diameter.
Hence, if two polygonal curves are near neighbors, then their vertices belong to the same sequence of at most $ k $ clusters, with good probability. In order to handle queries which are not known in advance, we consider cells, instead of clusters, which are subsets of the ambient space and are inherently computed by the same partition algorithm. 
We use perfect hashing and we build a hashtable where each non-empty bucket contains only those curves which intersect a certain sequence of cells. Now, any two curves which fall into the same bucket are $\Delta$-near, and by carefully adjusting the parameters, this already provides with an $\mathcal{O}(k)$ approximation. 
Furthermore, assuming the existence of a doubling oracle for the ambient space, we can  precompute answers to all query representatives defined by net points. To answer a query, we use  the net-hierarchy to efficiently compute the corresponding cell and then we retrieve the precomputed answer from a hashtable.
We present a weak version of our result for the black-box model in Theorem \ref{TannDFDhdDDweak}, and a stronger version for the  weakly explicit model in Theorem~\ref{TannDFDhdDD}.

\section{A sublinear distance oracle}\label{sec:distance:oracle}

We construct a data structure for distance queries to a curve in low dimension. Given a curve $X\in \MMM^m$, we compute an approximate curve simplification $X_k$, a set of curves $\mathcal{C}$ that we call \emph{query representatives}, and all distances between the query representatives and $X$. With these in memory, we can approximately compute $\df(X,q)$ for any query curve $q\in \MMM^k$.

\newcommand{\G}{\mathcal{G}}

\paragraph{Grid construction}\label{sec:gridconstruction:general}
First we compute the set of query representatives $\mathcal{C}$. 
We start by computing an $\alpha$-approximate curve simplification $X_k=p_1,\ldots,p_k$ for $X$.
Then we construct an exponential grid (for points) around every $p_i$. Exponential grids for points are known in the literature, we prove the following statement in Appendix~\ref{appendix:expgrid}) for completeness. 

\begin{restatable}{lemma}{expgridlemmacombined}\label{lem:expgrid:combined}
Let $(\MMM,\dm)$ be a metric space with constant doubling dimension $d_{\MMM}$. Assume the weakly explicit model. 
Let $x \in \MMM$, $\epsilon \in (0,1)$ and $r_1, r_2 \in \mathbb{R}$ with $r_1 < r_2$. Then  a set $\G(x) \subset M$ of size $\mathcal{O}((\log \frac{r_2}{r_1})\cdot\epsilon^{-d_{\MMM}})$ can be computed in time and space $\mathcal{O}(|\G(x)|)$ such that for every $y \in b_{\MMM}(x,r_2) \backslash b_{\MMM}(x,r_1)$, there is a point $z \in \G(x)$ with
\[
\dm(y,z) \le \epsilon \cdot \dm(x,y).
\]
If $(\MMM,\dm)$ is $(\mathbb{R}^d,||\cdot||)$, then it is possible to compute the nearest neighbor for a point $y \in b_{\MMM}(x,r_2)\backslash b_{\MMM}(x,r_1)$ in time $\log |\G(x)|$ (by grid snapping).
\end{restatable}

We call the union of all points in all exponential grids $\mathcal{G}$. Then we consider all curves of complexity $k$ that can be built by using points from $\mathcal{G}$, this set is $\mathcal{C}$. For each curve $X' \in \mathcal{C}$, we precompute the distance to the input curve $X$. We store all the distances in a data structure that allows to retrieve the distance $\df(X',X)$ using the sequence of vertices of $X'$ as a key. 
The following simple observation says that for every point on a query curve $q$ that is not too far from $X_k$ (with respect to Fr\'echet distance), we can shift every vertex to a point on one of the exponential grids. 
\begin{observation}\label{obs:a}
Let $p\in \MMM^m$ and $q=q_1,\ldots,q_m\in \MMM^m$. 
If $\df(p,q) \le \Delta$, then for all $j \in [m]$ there exists an $i \in [m]$ such that $q_j \in b_{\MMM}(p_i,\Delta)$.
\end{observation}
If the shifting is minor, then the Fr\'echet distance also does not change very much.
\begin{observation}\label{obs:b}
Let $q=x_1,\ldots,x_k$ and $q'=y_1,\ldots,y_k$ be curves in $\MMM^k$ assume that that $\dm(x_i,y_i) \le \delta$.
Then $|\df(q,Y) - \df(q',Y)| \le \delta$.
\end{observation}

With these in mind, we prove for the following lemma that for all $q$ which are sufficiently far away from $X$ but not too distant, $\mathcal{C}$ contains a good representative. The lower bound on the distance is necessary to make the exponential grids precise enough. The upper bound is to bound the size of the grids. 

\begin{restatable}{lemma}{lemmafrechetgrid}\label{lem:gridproperties}
For a curve $X \in \MMM^m$ and a precision parameter $\eps > 0$, it is possible to compute a set of query representatives $\mathcal{C}$ of size $\mathcal{O}(k^k\cdot (\log \frac{\alpha}{\epsilon})^k \cdot \left(\frac{\alpha}{\epsilon}\right)^{d_{\MMM} k})$ 
in time and space $\mathcal{O}(k\cdot|\mathcal{C}|+km)$ with the following property:
For any $q \in \MMM^k$ satisfying $\eps \cdot \df(X,X_k)\le \df(q,X) \le \frac{1}{\eps} \cdot \df(X,X_k)$, there is a curve $\bar{X}$ in $\mathcal{C}$ which satisfies that
\begin{align*}
\df(q,\bar{X}) \le \eps \cdot \df(q,X).
\end{align*}
In the Euclidean case, $\bar{X}$ can be found in time $\mathcal{O}(k +\log \frac{\alpha}{\epsilon})$ by grid snapping.
\end{restatable}
\begin{proof}
First we compute an $\alpha$-approximate curve simplification $X_k=p_1,\ldots,p_k$ for some constant $\alpha$, which needs time and space $\mathcal{O}(km)$ (see Theorem~\ref{thm:streamingeightapprox}).
Then for any query $q\in \MMM^k$ we know  that 
\begin{align}\label{eq:approx-curve-simplification}
\df(X,X_k) \le \alpha \cdot \df(q,X) 
\end{align}
is true. By the triangle inequality, this also implies
\begin{align}\label{eq:approx-curve-cor}
\df(q,X_k) \le \df(q,X) + \df(X,X_k) \le (1+\alpha) \cdot \df(q,X).
\end{align}
Next, we compute a set $\G(p_i)$ for all $i\in[k]$ by using Lemma~\ref{lem:expgrid:combined}. We use 
\[
r_1=\frac{\epsilon}{\alpha} \cdot \df(X,X_k) \stackrel{\eqref{eq:approx-curve-simplification}}{\le} \epsilon \cdot \df(q,X)
\] and $r_2=\frac{1}{\epsilon} \df(X,X_k)$, i.e., $r_2/r_1=\alpha \cdot \epsilon^{-2}$, and a precision parameter of $\epsilon'=\epsilon/(\alpha+1)$.
We union all resulting sets and $\{p_1,\ldots,p_k\}$ to obtain a set $\G$, i.e.,
\[
\G = \{p_1,\ldots,p_k\} \cup \bigcup_{i=1}^k \G(p_i).
\]
Notice that $\G$ is of size $\mathcal{O}(k\cdot (\log \frac{\alpha}{\epsilon}) \cdot (\frac{\alpha}{\epsilon})^{d_{\MMM}})$ for constant $d_{\MMM}$.

Now we obtain the set of query representatives $\mathcal{C}$ by constructing all curves that consist of $k$ points from $\G$. The cardinality of the resulting set of curves is $|G|^k$.

To show that the error bound holds, let $q=q_1,\ldots,q_k$ be any query curve that satisfies
$
\eps \cdot \df(X,X_k)\le \df(q,X) \le \frac{1}{\eps} \cdot \df(X,X_k)
$.
By Observation~\ref{obs:a}, every point $q_j$ of $q$ lies in  $b_{\MMM}(p_i,\frac{1}{\epsilon} \df(X,X_k))=b_{\MMM}(x_i,r_2)$ for at least one $i\in[k]$. For every $j\in[k]$, we set $x(q_j)$ to an $p_i$ such that $q_j \in b_{\MMM}(p_i,r_2)$ and such that $\dm(q_j,p_i)$ is minimal. 

Now we observe that 
\begin{align}\label{eq:yjxyj}
\dm(q_j,x(q_j))\le \df(q,X_k) \stackrel{\eqref{eq:approx-curve-cor}}{\le} (1+\alpha) \cdot \df(q,X).
\end{align}

Let $q_j$ be arbitrary. We map $q_j$ to a point $\pi(j)$.
If $q_j \in b_{\MMM}(x(q_j),r_1)$, then we set $\pi(j)=x(q_j)$. Thus, 
\[
\dm(q_j,\pi(j)) \le r_1 \le \epsilon\cdot \df(q,X).
\]
If $q_j \in b_{\MMM}(x(q_j),r_2)\backslash b_{\MMM}(x(q_j),r_1)$, then we set $\pi(j)$ to a nearest neighbor of $q_j$ in $\G_i$. Then 
\[
\dm(q_j,\pi(j)) \le \epsilon' \cdot \dm(q_j,x(q_j)) = \frac{\epsilon}{\alpha+1} \cdot \dm(q_j,x(q_j)) \stackrel{\eqref{eq:yjxyj}}{\le} \epsilon\cdot\df(q,X)
\]
 by Lemma~\ref{lem:expgrid:combined}.
Thus, $\dm(y_j,\pi(j)) \le \epsilon \cdot \df(q,X)$ is true
for all $j \in [k]$. By Observation~\ref{obs:b}, this implies that $\df(q,q')\le \epsilon \cdot \df(q,X)$. 
Since $q'\in\mathcal{C}$, the closest neighbor $X^\ast$ of $q$ in $\mathcal{C}$ satisfies $\df(q,X^\ast)\le \df(q,q')$, and that implies the statement of the lemma.

Finally, observe that if $\MMM^m=\XX^d_m$ for constant $d$, we can find $x(q_i)$ in time $\mathcal{O}(k)$, and can then find $\pi(q_i)$ by a case distinction. Then we find the layer of $\mathcal{G}(\pi(q_i))$ that $q_i$ lies in in time $\mathcal{O}(d)\subset \mathcal{O}(1)$, and finally compute the closest grid point by appropriate rounding in time $\mathcal{O}(\log |\mathcal{G}|) = \log ((\log \frac{\alpha}{\epsilon}) \cdot \frac{\alpha}{\epsilon}^{d}) \subset \mathcal{O}(\log \frac{\alpha}{\epsilon})$. The latter is possible because the grid in the layer is a regular and axis parallel grid.
\end{proof}
\paragraph{Query algorithm}

The query algorithm answers for queries $q\in \MMM^k$ either with $\df(q,X_k)+\df(X_k,X)$ or with $\df(q,X^\ast)+\df(X^\ast,X)$, where $X^\ast$ is either $\bar{X}$ from the above lemma or the nearest neighbor of $q$ in $\mathcal{C}$ in the metric case.
Which answer is returned depends on $\df(q,X_k)$: If $q$ is very close to $X_k$ (roughly at distance $\epsilon \cdot \df(X,X_k)$, but we have to account for the approximation factor of $X_k$ as well), then $X_k$ is a good representative of $q$ and we can return $\df(q,X_k)$. Similarly, if $q$ is very far away from $X_k$, then it is also very far away from $X$, and then $X$ and $X_k$ are very similar from the view point of $q$. Then $\df(q,X_k)$ is a good approximation. For the intermediate case, we have the exponential grids, and have shown that $\mathcal{C}$ contains a good representative for $q$. 
The following query algorithm follows along these lines; that it always returns a sum of two distances is to ensure that we always overestimate the distance and never underestimate it.

\algorithmus{\texttt{query}($q \in \MMM^k$, $X_k\in \MMM^k$, $\mathcal{C}$)}{
\begin{compactenum}
  \setlength{\itemsep}{0.0cm}
  \item {\bf If} $\df(q,X_k) \le \frac{\eps}{2\alpha} \cdot \df(X,X_k)$ {\bf or} $\df(q,X_k) \ge \frac{2(\alpha+1)}{\eps}\cdot \df(X,X_k)$
  \item \hspace*{0.5cm} Set $X^\ast := X_k$
  \item {\bf Else} 
  \item \hspace*{0.5cm} Set $X^\ast:=$\texttt{NearestGridNeighbor}($\mathcal{C},q$) \quad\quad (or $X^\ast := \bar{X}$ if $\MMM^m=\XX_d^m$)
  \item {\bf Return } $\Delta = \df(q,X^\ast)+\df(X^\ast,X)$
\end{compactenum}
}
In the Euclidean case, we find the nearest neighbor in time $\mathcal{O}(k + \log \frac{\alpha}{\epsilon})$ according to Lemma~\ref{lem:gridproperties}.
In the metric case, we use an ANN data structure to store the \lq grid\rq\ points, which gives a slightly slower query time. 
After finding $X^\ast$, we retrieve the value of $\df(X^\ast,X)$ from the data structure storing the precomputed distances. 

\begin{restatable}{lemma}{lemQueryAlg}
Let $X \in \MMM^m$, $X_k\in \MMM^k$ an $\alpha$-approximate curve simplification for $X$, $\eps > 0$ and $\mathcal{C}$ be a set of query representatives computed according to Lemma~\ref{lem:gridproperties} with precision parameter $\eps/2$.
Then the value $\Delta$ computed by \texttt{query}($q$,$X_k$, $\mathcal{C}$) satisfies 
\[
\df(q,X) \le \Delta \le (1+\epsilon)\cdot \df(q,X).
\]
\end{restatable}

\begin{proof}
We consider three cases.
If $\df(q,X_k) \le \frac{\eps}{2\alpha} \cdot \df(X,X_k)$ then we get by the triangle inequality that
\begin{align*}
\Delta = \df(X,X_k)+ \df(q,X_k)\le& 2\df(q,X_k) + \df(q,X) \\
\le& \df(q,X) + 2\frac{\eps}{2\alpha} \cdot \df(X,X_k) 
\stackrel{\eqref{eq:approx-curve-simplification}}{\le} (1+\eps)\cdot \df(q,X).
\end{align*}
Intuitively, $\df(X,X_k)$ dominates our distance estimation in this case, since $\df(q,X_k)$ is very small. By returning the sum of the two, we make sure that we never underestimate the true distance: The triangle inequality implies that
\[
\Delta = \df(X,X_k)+ \df(q,X_k) \ge \df(q,X) - \df(q,X_k) + \df(q,X_k) = \df(q,X).
\]
Similarly, if $\df(q,X_k) \ge \frac{2(\alpha+1)}{\eps}\cdot \df(X,X_k)$ then we observe that 
\begin{align*}
   \Delta=\ &  \df(q,X_k) + \df(X,X_k)
\le\  \df(q,X) + 2\df(X,X_k) \\
\le\ & \df(q,X) + 2\frac{\eps}{2(\alpha+1)}\cdot \df(q,X_k)
\stackrel{\eqref{eq:approx-curve-cor}}{\le}\ 
      (1+\eps) \cdot \df(q,X)
\end{align*}
and that
\[
\Delta =  \df(q,X_k) + \df(X,X_k)
\ge \df(q,X) - \df(X,X_k) + \df(X,X_k) = \df(q,X).
\]
In the remaining case, we use the guarantee from Lemma~\ref{lem:gridproperties}, which says that $\df(q,X^\ast) \le \frac{\eps}{2} \cdot \df(q,X)$. This implies by the triangle inequality that
\begin{align*}
    \Delta = \df(q,X^\ast)+\df(X^\ast,X)
    \le 2\df(q,X^\ast) + \df(q,X)
    \le (1+\eps)\cdot \df(q,X).
\end{align*}
Furthermore, we do not underestimate the distance since
\[
\Delta = \df(q,X^\ast) + \df(X^\ast,X) \ge \df(q,X) - \df(X^\ast,X) + \df(X^\ast,X) = \df(q,X),
\]
and that concludes the proof.
\end{proof}

\paragraph{Result}

We collect our findings and summarize the precomputation and query times as well as the memory need in the following theorem. The main missing ingredient discussed there (aside from carefully looking at all steps) is the (approximate) nearest neighbor data structure we use for the metric case, where we use a result by Arya et. al.~\cite{AMVX08}.

\begin{theorem}[\cite{AMVX08}]
\label{theorem:arya}
Let $S$ be a set of $n$ points in a metric space $\MMM$ in the weakly explicit model with doubling dimension $d_{\MMM}$. Let $0 \le \epsilon \le 1/2$ and $2 \le \gamma \le 1/\epsilon$. It is possible to construct a data structure with space $n \gamma^{\mathcal{O}(d_{\MMM})} \log(1\backslash \epsilon)$ that can answer $\epsilon$-approximate nearest neighbor queries in time $\mathcal{O}(\log(n \gamma))+(1/(\epsilon\gamma))^{\mathcal{O}(d_{\MMM})}$. The time to construct the data structure is $n (1/\epsilon)^{d_{\MMM}} \log n$.
\end{theorem}

\thmTimeSpace*
\begin{proof}

Computing $\mathcal{C}$ takes time and space $\mathcal{O}(k \cdot |\mathcal{C}|+km)$ by Lemma~\ref{lem:gridproperties}. 
Recall that $\mathcal{C}$ results from finding all curves with $k$ points from a $\mathcal{G}$, so $|\mathcal{C}|=|\mathcal{G}|^k$, and we know that $|\mathcal{G}| = k \cdot \frac{\alpha}{\epsilon} \cdot \log \frac{\alpha}{\epsilon}$.

The distance precomputations between any curve in $\mathcal{C}$ and $X$ can be done in time and space $\mathcal{O}(mdk)$ using dynamic programming~\cite{eiter1994computing}, so we get a total precomputation time of $\mathcal{O}(mk + |\mathcal{C}| \cdot mk)$ (distance computations are assumed to take $\mathcal{O}(1)$).
We store all precomputed distances, plus $X_k$ and $\df(X,X_k)$, requiring $\mathcal{O}(|\mathcal{C}|+k)\in \mathcal{O}(f_s(k,\epsilon,\alpha))$ space.  

For a query curve $q \in \MMM^k$, we need to compute the Fr\'echet distance between $q$ and $X_k$, which takes time $\mathcal{O}(k^2)$, and to find $X^{\ast}$.
In the Euclidean case, we do the latter in time $\mathcal{O}(k+\log\frac{\alpha}{\epsilon})$ by Lemma~\ref{lem:gridproperties}.

In the general metric case, we use the ANN nearest neighbor data structure by Arya et. al.~\cite{AMVX08}, and we invoke Theorem \ref{theorem:arya}. 
We build the data structure with $\epsilon=1/2$ and $\gamma=2$ on the $|\mathcal{G}|= \mathcal{O}(k \cdot \frac{\alpha}{\epsilon} \log \frac{\alpha}{\epsilon})$ points in $\mathcal{G}$. 
The data structure thus answers with a constant factor approximation, but we compensate that by replacing $\epsilon$ in our construction by $\epsilon'=(2/3)\cdot \epsilon$ in the end.  

The ANN data structure is of size $\mathcal{O}(|G| \cdot  2^{\mathcal{O}(d_{\MMM})} \cdot \log_2 2 ) \subset \mathcal{O}(|G|)$, the query time for one point is $\mathcal{O}(\log |G|\cdot 2+(1/(2/2)^{\mathcal{O}(\d_{\MMM})}) \subset \mathcal{O}(\log |G|)$, and the preprocessing time is $\mathcal{O}(|G| \cdot 2^{\mathcal{O}(d_{\MMM})} \cdot \log |G|) \subset \mathcal{O}(|\mathcal{C}|)$.
The additional space requirement and the additional precomputation times are dominated by the space and time for the rest of the data structure. The query time becomes $\mathcal{O}(k^2+k \cdot \log |G|) \subset \mathcal{O}(k^2 + k \cdot \log (k \cdot \frac{\alpha}{\epsilon}))$.

Finally, notice that if we want to answer queries for curves $q\in \MMM^{k'}$ for $k' \le k$, we can artificially extend them to $k$ vertices by repeating the last vertex $k-k'$ times. Thus, we can use the data structure for all curves up to length $k$.

\end{proof}

\section{A distance oracle for the streaming setting}\label{sec:streaming}

Let $\DD_k^{\beta}(X)$ denote a data structure built on input $X$ which uses space in  $\mathcal{O}(f_s(k,\eps,\alpha))$  and which can answer $\beta$-approximate queries with query curves  $q \in \XX^d_k$ for the discrete Fr\'echet distance $\d_{dF}(q,X)$ in  $\mathcal{O}(f_q(k,\eps,\alpha))$ time. We studied building such data structures in Section~\ref{sec:distance:oracle}, see Theorem~\ref{thm:distoracle}. In the following lemmas we study composition properties of such data structures.

\subsection{Composition properties}\label{sec:composition}

We show that given multiple data structures built on individual sequences, we can answer queries for the concatenation of the sequences. For two individual sequences, we can do this by finding the splitting point on the query curve where an optimal traversal switches from pairing with the first sequence to pairing with the second sequence. For multiple sequences and distance oracles, we could apply this recursively. Naively, this would take roughly $k^{\ell}$ calls to the composition lemma. We avoid this by using dynamic programming. 

\begin{lemma}\label{lem:composition:chain}
Assume we are given a sequence of data structures $\DD_1,\dots,\DD_{\ell}$, where $\DD_i = \DD_k^{\beta}(X_i)$ for a sequence $X_i$, and we wish to answer distance queries on the concatenated sequence $X=X_1\circ X_2 \circ \dots \circ X_{\ell}$.
For any query curve $q \in \XX^d_k$, and without direct access to $X_i$ for $1 \leq i \leq \ell$, we can output a value $d'(q)$ in query time $\mathcal{O}( (k^2 \cdot (\ell-2) + k) \cdot f_q(k,\eps,\alpha))$ such that
\[\df(q, X) \le d'(q) \le \beta \cdot \df(q,X). \]
This can be done while using $\mathcal{O}(k \ell)$ additional space.
\end{lemma}

\begin{proof}
We use dynamic programming on all subsequences of $q$ paired with the $X_j$'s. We denote with $q_{[a,b]}$ the subsequence $q_a,\dots,q_b$ of $q$.
By the definition of the discrete Fr\'echet distance it holds that
\[ d(q_{[1,k']}, X_1 \circ \dots \circ X_{\ell'}) = \min_{ 1\leq i \leq k' \atop b \in \{0,1\} } \max ( d(q_{[1,i]}, X_1 \circ \dots \circ X_{\ell'-1}), 
d(q_{[i+b,k']}, X_{\ell'}). \] 

Our dynamic programming scheme stores partial solutions in a $k \times \ell$ matrix. That is, in the cell with index $(k',\ell')$ we store a $\beta$-approximate estimate to $ d(q_{[1,k']}, X_1 \circ \dots \circ X_{\ell'}).$ 
In the recursive statement above, there are $2k'$ possible assignments to $i$ and $b$. For each such assignment to $i$ and $b$, we can compute a $\beta$-approximation to  $d(q_{[i+b,k']}, X_{\ell'})$ by querying $\DD_{\ell'}$ and use recursion to estimate the term
$ d(q_{[1,i]}, X_1 \circ \dots \circ X_{\ell'-1})$. Note that for $\ell'=1$ we obtain the term
$d(q_{[1,k']}, X_1 )$, for which we can compute a $\beta$-approximation by using a query to $\DD_1$.
We can also avoid the recursion and fill the table in a bottom-up fashion starting with $\ell'=1$. There are $k\cdot \ell$ cells and for each cell we issue $\mathcal{O}(k)$ queries to one of the data structures. Note that, technically speaking, for the cells with index $\ell'=1$ and the cells with index $\ell'=\ell$ we need to issue one query each only, which gives us a slightly better query time if $\ell=2$. In total, we have $\mathcal{O}(k^2 (\ell-2) + k)$ queries, which each take time in $\mathcal{O}(f_q(k,\eps,\alpha))$.
\end{proof}

Lemma~\ref{lem:composition:chain} shows that one can use multiple data structures built on individual sequences to answer queries on the concatenation of the sequences. The following lemma shows that, if we have access to a $k$-simplification of the concatenated sequence, we can \emph{compose} two data structures and obtain a new data structure with the same properties with slight increase in the dependency on $\eps$, but the same dependency on $k$. This will be crucial for applying the merge-and-reduce framework.

\begin{lemma}\label{lem:composition:build:chain}
Given $\ell$ data structures $\DD_1,\dots,\DD_{\ell}$, where $\DD_i = \DD_k^{\beta}(X_i)$ for a sequence $X_i$, and an $\alpha$-approximate $k$-simplification of the concatenated sequence $X=X_1\circ\dots\circ X_{\ell}$. 
Without direct access to $X$, it is possible to compute a data structure in time $\mathcal{O}( f_s(k,\epsilon,\alpha) \cdot (k^2 \cdot (\ell-2) + k) \cdot f_q(k,\eps,\alpha))$ and space $\mathcal{O}(f_s(k,\epsilon,\alpha))$  that satisfies the following:
For any query curve $q \in \Rdk$, the data structure outputs a value $d'(q)$ in query time $\mathcal{O}(f_q(k,\eps,\alpha))$ such that 
\[\df(q,X) \le d'(q) \le \beta^2\cdot \df(q,X).\]
\end{lemma}

\begin{proof}
This follows along the lines of the proof of Theorem~\ref{thm:distoracle}, the main theorem from the previous section. The only difference is that we estimate the distance of a grid curve to the curve $X$ using Lemma~\ref{lem:composition:chain} in $\mathcal{O}((k^2 \cdot (\ell-2) + k) \cdot f_q(k,\eps,\alpha))$ time. This incurs an additional $\beta$-factor in the approximation quality.
\end{proof}

\subsection{Approximate curve simplification in the streaming setting}\label{sec:simplification}
In this section, we develop a streaming algorithm for the approximate curve simplification problem. 
More precisely, let $(\MMM,\dm)$ be a metric space, and assume that we have oracle access to the distance function, i.e., we can query $\dm(x,y)$ for any $x,y \in M$ in constant time. During this section, we abbreviate $d(x,y)=\dm(x,y)$ for brevity.

We are given a curve $X=p_1,\ldots,p_m$ as a stream, i.e., the vertices are given in the order that they have on the curve. We compute an approximate curve simplification while never storing more than $\mathcal{O}(k)$ points.

The algorithm is inspired by a streaming algorithm for $k$-center which is due to  Charikar et. al.~\cite{CCFM04}. 
The key idea of our adapted algorithm is to maintain a value $\ell$ which satisfies three properties: 
After processing the prefix $X'$ of $X$ and computing a simplification $X_k$, all vertices on $X_k$ have a distance of at least $\ell$ to their neighbors, and
we have $\df(X,X_k) \le \alpha \cdot \ell$ for some constant $\alpha > 0$, and $\df(X,O_k)\ge \ell / \beta$ for some constant $\beta$ and an optimal curve simplification $O_k$. 

Before we describe the details of this algorithm, we define a bit more notation.
For any curve $Y=y_1,\ldots,y_m$ and any $i\in \{2,\ldots,m-1\}$, we call $y_{i-1}$ the \emph{predecessor} of $y_i$ and $y_{i+1}$ the \emph{successor} of $y_i$. 
Both predecessor and successor are also called \emph{neighbor} of $y_i$.

During the algorithm, we maintain a current curve simplification $X_k$ for the prefix $X'$ of $X$ that we have read so far, but we \emph{additionally} implicitly also compute a traversal of $X'$ and $X_k$. We mainly need this traversal for the purpose of analyzing the quality of our solution. 
What we maintain is the following: 
Say that we have read a prefix $X'=p_1,\ldots,p_{m'}$ of length $m'$. 
For every vertex $q_i$ of $X_k=q_1,\ldots,q_k$, we have an interval $[a_i,b_i]\subset \{1,\ldots,m'\}$, such that for all $i, j \in [k]$, $i<j$, $b_i < a_j$ and such that the union of all intervals is just $[m']$. 
We call this an \emph{interval partitioning} of $X'$. 
The implicit traversal can be constructed by starting at $q_1$ on $X_k$ and traversing all vertices $p_i, i \in [a_1,b_1]$, then jumping to $q_2$ and $a_2$, traversing all (remaining) vertices $p_i, i \in [a_2,b_2]$, and so on. 
We denote the subcurve $p_{a_i},\ldots,p_{b_i}$ of $X'$ by $X_{a_i b_i}$.

We describe our algorithm in two parts. For the initial part of the stream, we use the following routine.

\algorithmus{\texttt{init}($k \in \mathbb{N}$)}{
\begin{compactenum}
    \item Read the first $k+1$ vertices of $X$, if a vertex is identical to its predecessor, ignore it
    \item Let $Y=y_1,\ldots,y_{k+1}$ be the resulting curve (it has $k+1$ vertices, and every vertex is different from its neighbors)
    \item Let $d = \min_{i\in\{1,\ldots,k\}} d(y_i,y_{i+1})$ be the smallest distance between two neighbors on Y
    \item Choose $i \in [k]$ such that $d=d(y_i,y_{i+1})$ and remove $y_i$ from $Y$
    \item Call the resulting curve $X_k = q_1,\ldots,q_k$ and set $\ell=d$
    \item Set $a_j=b_j=j$ for all $j < i$, $a_i = i$ and $b_i=i+1$, and $a_j=b_j=j+1$ for all $j > i$
\end{compactenum}
}

\begin{observation}
The curve $X_k$ computed by \texttt{init} is a $2$-approximate curve simplification for the prefix $X'$ of $X$ that was processed.
\end{observation}
\begin{proof}
First observe that ignoring vertices that are identical to their predecessor does not incur any error: In any traversal, we traverse these identical neighbors of a point $p_i$ right after traversing $p_i$. Thus, $\df(Y,X')=0$.
The Fr\'echet distance between $X'$ and $X_k$ is $d$: We can traverse both $y_{i'}$ and $y_{i'+1}$ on $X'$ while staying at $y_{i'}$ on $X_k$, where $i'$ is the index of the ignored vertex. On the other hand, an optimal simplification $O_k$ with $k$ vertices has to traverse two of the neighbors on $X'$ while staying at one vertex of the simplification, since $X'$ has one vertex more than $O_k$. Thus, $\df(O_k,X') \ge d/2$.
\end{proof}

After the curve $X_k$ is initialized, we call the following update routine for each further point $p_i$ on $X$.

\algorithmus{\texttt{update}($p \in \MMM$)}{
\begin{compactenum}
    \item {\bf If} $d(p,q_k)\le \ell$ {\bf Then} do nothing 
    \item {\bf Else} set $q_{|X_k|+1}=p$, append it to $X_k$, and set its interval to $[|X'|,|X'|]$
    \item {\bf While} $X_k$ has $k+1$ vertices {\bf do}
    \item \hspace*{0.5cm} Set $i=1$
    \item \hspace*{0.5cm} {\bf  While} $i < k+1$ {\bf do}
    \item \hspace*{1cm} Set $j=i+1$
    \item \hspace*{1cm} {\bf While} $j < k+1$ and $d(q_i,q_{j}) \le 2\ell$
    \item \hspace*{1.5cm} Delete $q_j$ and set $b_i = b_j$
    \item \hspace*{1.5cm} Set $j=j+1$
    \item \hspace*{1cm} Set $i=j$ and rename the vertices such that $X_k=(q_i)_i$
    \item \hspace*{1cm}  Set $\ell = 2 \cdot \ell$
\end{compactenum}
}

\begin{lemma}
After each call to $\texttt{update}$, the following invariants are true for the prefix $X'$ read from $X$ and the current simplification $X_k=(q_i)_1^{m'}$:
\begin{enumerate}
    \item All vertices on $X_k$ have a pairwise distance of at least $\ell$ to their neighbors.
    \item $\df(X',O_k) \ge \ell/4$, where $X'$ is the processed part of the input curve, and $O_k$ is an optimal curve simplification for $X'$ with $k$ vertices.
    \item For all $i \in [m']$, and all $j \in [a_i,b_i]$,
    $\df(q_i,q_j) \le 2 \cdot \ell$, and thus, $\df(X',X_k)\le 2 \cdot \ell$ holds.
\end{enumerate}
\end{lemma}
\begin{proof}
All invariants are true before the first call to \texttt{update}: $\ell$ is set to the smallest distance between a vertex and a neighbor, so invariant $1$  and $2$ are true, and since only one vertex that is at distance $\ell$ from its neighbor is missing on $X_k$, we have for all $i \in [k]$, $\df(X_{a_i b_i},q_i) \le \ell$ (invariant $3$).

Now assume that all invariants are true and $\texttt{update}$ is performed. In line $2$, a vertex is appended to $X_k$ if and only if its distance to the currently last vertex on $X_k$ is larger than $\ell$. If $d(p,q_k) \le \ell$, then ignoring the vertex does not violate any invariant: Traversing it can be appended to the current traversal, and this does not increase the traversal cost above $\ell$. If $d(p,q_k) > \ell$, then adding the vertex does not violate invariant $1$ since the point is at distance at least $\ell$ from its predecessor.
Invariant $2$ is not affected since $\ell$ is not (yet) changed and $\df(X',O_k)$ can only increase. 
Invariant $3$ stays true as well: We do not change any of the existing vertices $q_i$ nor their intervals $[a_i,b_i]$, and the new vertex' interval is the point itself (so the Fr\'echet distance is zero for this subcurve). 

Now we show that the invariants stay true after each iteration of the outer while-loop. One iteration of the while-loop does the following: Starting with a curve with $k+1$ vertices where every vertex has distance $\ge \ell$ to its neighbor(s), it removes vertices such that every removed vertex is at distance $\le 2 \ell$ of a non-removed neighbor, and such that the distance between any vertex and its neighbor(s) is $2 \ell$ afterwards. (Observe that one iteration of the while loop may do nothing if the distances between the neighbors are actually already higher than $2\ell$; this is why the while loop continues until it finally has removed at least one vertex.) As a last step, $\ell$ is set to $2 \ell$. 

Invariant $1$ is true after this, since the algorithm removed all neighbors that are at distance $\le 2\ell$ from their neighbor and in the end doubles $\ell$.

For invariant $2$, we observe that when we an iteration of the outer while loop, we have $k+1$ vertices on our curve, and each vertex is at distance $\ge \ell$ by the induction hypothesis (invariant $1$). Any optimal curve simplification $O_k$ has to traverse two neighbors while staying on the same vertex at $X_k$. This means that $\df(X',X_k) \ge \ell/2$, or, after we set $\ell = 2\ell$, $\df(X',X_k)\ge \ell/4$.

For invariant $3$, we observe that by invariant $3$ from the induction hypothesis, we had an interval partitioning satisfying $d(q_i,q_j)\le 2 \cdot \ell$ for all $i \in [k+1]$ and all $j \in [a_i,b_i]$, where $X_k=(q_i)_{i=1}^{m'}$ is the curve before the iteration of the outer while loop. Observe that we only merge intervals, so our new interval partitioning still covers $X'$.
Let $(q_i')_{i=1}^{m'}$ be the new curve after the while loop, and fix an $i \in [m']$. Call the new interval of $q_i'$ $[a_i',b_i']$, and let $j \in [k+1]$ be the index of a vertex $q_i$ that was removed and is now in the interval $[a_i',b_i']$. This happened because $d(q_i,q_i')\le 2 \ell$.
We conclude that $d(q_i',q_j)\le d(q_i',q_i) + d(q_i,q_j) \le 4\ell$ is true for all $j \in [a_i,b_i]$ (the old interval of $q_i$). Thus, every vertex in $[a_i',b_i']$ is at distance at most $4 \ell$ from $q_i'$. Since $\ell$ is doubled at the end of the iteration of the outer while loop, invariant $3$ is again true.
\end{proof}

\begin{corollary}
After each call to \texttt{update}, the current simplification $X_k$ is an $8$-approximate curve simplification for the prefix $X'$ of $X$ that was processed. The storage requirements during the initialization and during the calls to \texttt{update} is $\mathcal{O}(k)$, assuming that a point from the metric space can be stored in constant space.
\end{corollary}

The pseudo-code for \texttt{update} is not efficient in its above form. The outer while-loop is convenient for the purpose of the above proof, however, it may iterate many times until $\ell$ has been sufficiently increased. In an efficient implementation, we compute the first $i$ such that $2^i \ell$ is at least the smallest pairwise distance between two neighbors on $X_k$ and directly jump to the corresponding iteration of the while loop. This version of \texttt{update} thus needs only one pass over $X_k$ to reduce its cardinality to $\le k$. That means that one call of \texttt{update} has a running time in $\mathcal{O}(k)$, and processing a stream with $n$ vertices has a total running time of $\mathcal{O}(mk)$. 

\thmstreamingsimplification*

\subsection{The main algorithm---putting everything together}
\label{sec:streamingDistOracle}
\label{sec:streaming:algo}

We now describe our streaming algorithm for computing the distance oracle in a stream. We use the composition lemmas in Section~\ref{sec:composition}, and the streaming algorithm for computing simplifications in the stream from Section~\ref{sec:simplification} . 

The algorithm reads the input stream in blocks of size $k$ (the content of each block is stored in a buffer until the end of the block has been reached). The data structure of Theorem~\ref{thm:distoracle} is built for the sequence in a block when it has been read. The algorithm maintains a stack $S$ with the data structures computed and maintained so far.  In addition, the algorithm maintains a dynamic search structure $A$ to store $k$-simplifications of the input stream starting at selected points in the stream.  In $A$, the $k$-simplification starting at the beginning of the $i$th block is stored with key $i$. The simplifications are updated with every point that is read from the stream, by using Theorem~\ref{thm:streamingeightapprox}. 
We want the number of such elements in $S$ not to become too large. To this end, we need to reduce memory by repeatedly merging and reducing the data structures in $S$. 
The basic merge-reduce operation pops the $t$ topmost data structures from the stack, applies Lemma~\ref{lem:composition:build:chain} and pushes the result back to the stack. 

\algorithmus{\texttt{merge-reduce}}{
\begin{compactenum}
    \item {\bf For} $1 \leq i \leq t$
    \item \hspace{0.5cm}  $(\DD_i,j_i) \leftarrow pop(S)$
    \item $p \leftarrow find(A,j_t)$ 
    \item Use Lemma~\ref{lem:composition:build:chain} with $\DD_t$,\dots, $\DD_1$ and $k$-simplification $p$ to build new data structure $\DD'$
    \item $push(S, (\DD',j_t))$ 
    \item {\bf For} $1 \leq i \leq t-1$
    \item \hspace{0.5cm} $delete(A,j_i)$ 
\end{compactenum}
}

Each call to \texttt{merge-reduce} incurs an extra approximation factor on the data structures newly produced. Therefore, we need to be careful in implementing a scheme which bounds the sequence of reductions leading to any newly created data structure. This is done as follows. (We adopt the convention that indices start counting at $1$, i.e., the first block read by the algorithm has index $i=1$.)

\algorithmus{\texttt{read($x,l$)}}{
\begin{compactenum}
    \item Let $i:=\lceil l/b \rceil$ denote the index of the current block
    \item {\bf If} $(l\bmod{b} == 1)$ and $(i\bmod{t} == 1)$: 
    \item \hspace*{0.5cm} Initialize an empty $k$-simplification and store in $A$ with key value $i$
    \item Append $x$ to buffer $B$ of current block
    \item Update all simplifications stored in $A$ using Theorem~\ref{thm:streamingeightapprox} with new point $x$
    \item {\bf If} $(l\bmod{b}==0)$:
    \item \hspace*{0.5cm} \texttt{process($B, i$)}
    \item \hspace*{0.5cm} Flush buffer $B$
\end{compactenum}
}

\algorithmus{\texttt{process($B,i$)}}{
\begin{compactenum}
    \item Use Theorem~\ref{thm:distoracle} to build data structure $\DD$ on buffered data $B$
    \item $push(S, (\DD,i))$
    \item $j = i$
    \item {\bf While} $(j\bmod{t} == 0)$: 
    \item \hspace*{0.5cm} Apply \texttt{merge-reduce} operation to $S$
    \item \hspace*{0.5cm} $j \leftarrow j/t $
\end{compactenum}
}

\emph{Answering a query.}
At any point in time, we can answer distance queries on $X$, ending at any point in the current block, by using Lemma~\ref{lem:composition:chain} on the data structures stored in $S$. For this purpose we read all data structures stored in $S$, but leave $S$ otherwise unchanged. To answer a query to a prefix curve that ends in the middle of the current block, we  simulate the respective queries to the data structure on the partial block which are done in the query algorithm of Lemma~\ref{lem:composition:chain} by using the exact algorithm~\cite{eiter1994computing}. For simplicity of presentation, we assume a block size of $k$, so that these queries can be performed within the time bounds of Theorem~\ref{thm:distoracle}.

\subsection{Analysis}
\label{sec:streaming:analysis}

We first prove some structural lemmas before we give the main theorem. Let $\DD_{[u,v]}$ denote a data structure created by the streaming algorithm which answers distance queries for the range of blocks $[u,v]$.

\begin{observation}\label{obs:stack:order}
Let $w$ be the number of data structures stored in $S$ after  execution of \texttt{read} on block $i$ and let $\DD_{[u_1,v_1]}, \dots, \DD_{[u_{w},v_{w}]}$ denote the data structures stored in the stack. It holds that $u_1=1$, $v_{w}=i$, and $v_i=u_{i-1}+1$. Moreover, it must be that $i=v_w$.
\end{observation}

\emph{Merge tree.} Consider the set of data structures created by the algorithm  and consider a tree, where every node corresponds to one of such data structures and an edge of the tree is drawn between a data structure resulting from a call to Lemma~\ref{lem:composition:build:chain}  and each of the data structures directly involved in the call. We call this tree the \emph{merge tree}.
Lemma~\ref{lem:t:merge:tree} below implies that every data structure residing in the stack is the result of a perfectly balanced $t$-ary merge tree on a sequence of consecutive blocks of the input stream. More specifically, the leaves correspond to data structures created by a call to Theorem~\ref{thm:distoracle} and each inner node corresponds to the data structure created by a call to Lemma~\ref{lem:composition:build:chain} on its children.

\begin{restatable}{lemma}{lemmamergetree}\label{lem:t:merge:tree}
Let $u \leq v$ be two block indices. 
If there exists some value $\ell \geq 1$, such that $v = u + t^{\ell} - 1$ and $(v \bmod{t^{\ell}} == 0)$,
then, after processing block $v$, a data structure $\DD_{[u,v]}$ has been created by the streaming algorithm as a result of applying Lemma~\ref{lem:composition:build:chain} to $t$ data structures with ranges $[u_1,v_1], \dots, [u_t,v_t]$, where $u_1=u$, $v_t=v$, $u_i = u + (i-1) \cdot t^{\ell-1}$, and $v_i=u_i + t^{\ell-1} - 1$.
If $u=v$, then a data structure $\DD_{[u,v]}$ was created by the streaming algorithm as a result of applying Theorem~\ref{thm:distoracle} to block $v$.
\end{restatable}\begin{proof}
We use an induction on increasing value of $v$. For the base case, let $v=1$. In this case we have $u=v=1$, and $\DD_{[u,v]}$ is the data structure created by the algorithm on the first block of the stream. This satisfies the claim.

In the induction step, we perform an inner induction with decreasing index $u$. The base case of this inner induction is $u=v$. In this case, the streaming algorithm creates a data structure $\DD_{[u,v]}$ on block~$v$. 
Again, this satisfies the claim. Now fix a value $u < v$. 
Assume there exists some $\ell\geq 1$, such that $(v \bmod{t^{\ell}} == 0)$ and $v = u + t^{\ell} - 1$, as specified in the lemma (otherwise, the claim is satisfied). 
Consider a range $[u_i,v_i]$ with 
 $u_i=u + (i-1) \cdot t^{\ell-1}$ and
$v_i=u_i + t^{\ell-1} - 1$, for some $1\leq i \leq t$. Note that either $u_i < v_i < v$ or $u < u_i < v_i \leq v$. In either case we can apply the induction hypothesis to the pair of values $u_i,v_i$. Now, assume $\ell=1$, then, by definition,
$ v_i = u_i+t^{\ell-1}-1 =u_i$ and therefore the streaming algorithm has created a data structure $\DD_{[u_i,v_i]}$ at a leaf of the merge tree.
Otherwise, if $\ell>1$, then there exists an $\ell'=\ell-1$ with $ v_i = u_i+t^{\ell'}-1$ and $(v_i \bmod{t^{\ell'}}==0)$. By induction, we conclude that the streaming algorithm has created a data structure $\DD_{[u_i,v_i]}$. 

By Observation~\ref{obs:stack:order}, these data structures 
$\DD_{[u_1,v_1]},\dots,\DD_{[u_t,v_t]}$
are stored at some point as the $t$ topmost data structures in $S$, when processing the block $v$, since $v=v_t$. By our initial assumption we have $(v \bmod{t^{\ell}}==0)$. Therefore, the algorithm will create a new data structure $\DD_{[u,v]}$ as a result of using Lemma~\ref{lem:composition:build:chain} with $\DD_{[u_1,v_1]},\dots,\DD_{[u_t,v_t]}$, the $t$ topmost the data structures on the stack. This proves the claim.
\end{proof}

\begin{lemma}
\label{lem:t:merge:depth}
After processing block $i$: 
\begin{compactenum}[(i)]
    \item The height of the merge tree of any data structure created is at most $\lceil\log_t i\rceil$.
    \item The total number of calls to \texttt{merge-reduce} performed  is in $\mathcal{O}(i)$.
    \item The number of elements in $S$ is at most $(t-1) \cdot \lceil\log_t i \rceil$
    \item The number of elements in $A$ is at most  $\lceil\log_t i \rceil$
\end{compactenum}
\end{lemma}

\begin{proof}
First assume that $i=t^{\ell}$ for some integer value of $\ell$. By Lemma~\ref{lem:t:merge:tree}, the algorithm has created a data structure $\DD_{[1,i]}$ and its merge tree contains all data structures created by the algorithm up to this point. The height of this merge tree is $\ell = \log_t i$ and the number of inner nodes is $\sum_{j=1}^\ell \frac{i}{t^{j}}$, which is linear in $i$.
Now the claims (i) and (ii) follow from Lemma~\ref{lem:t:merge:tree}. In particular, with respect to (i), every data structure created is a node of the merge tree and its height is at most the height of the tree. With respect to (ii), each call to \texttt{merge-reduce} corresponds to an inner node of the merge tree.

Now, assume that $i$ is not a power of $t$ but consider the merge tree of $i'$, where for some integer value of $\ell$: $t^{\ell-1} < i < t^{\ell} = i'$ (that is, $i'$ is the next power of $t$). Claims (i) and (ii) follow from the above, since $\log_t i' = \lceil \log_t i \rceil$.
As for (iii), this also follows from Lemma~\ref{lem:t:merge:tree}, since the data structures stored in $S$ correspond to the subtrees hanging to the left the path from the root of the merge tree to the index of the current block. Each node on this path has at most $(t-1)$ children to the left of the  path.
Finally, to show (iv), note that after completion of any call to \texttt{process}, the number of elements in $A$ is the number of elements in $S$ divided by $(t-1)$, as can be shown by induction on the blocks in the order in which they are processed in the stream.
\end{proof}

In the following main theorem of our analysis the parameter $s$ provides a trade-off between the number of data structures maintained by the streaming algorithm and the space used by the individual data structures, which depends on the function $f_s(k,\eps,\alpha)$. 

\begin{theorem}
\label{thm:streamingDistOracleTradeoff}
For any value of $1 < s \leq \log_2 m$ there is a streaming algorithm that computes a distance oracle for a curve $X$ given as a stream with $m=|X|$, where $m$ is known in advance, with the following properties. The algorithm uses additional memory in 
$ \mathcal{O}(s \cdot m^{1/s} \cdot f_s(k, \frac{\epsilon}{s},\alpha))$ and the total time used by the algorithm to compute a distance oracle in the stream is bounded by 
\[
\mathcal{O}\left( m^{(1+\frac{1}{s})} \left(s +  k \cdot f_s\left(k,\frac{\eps}{s},\alpha\right) \cdot  f_q\left(k,\frac{\eps}{s},\alpha\right)\right) \right)
\]
Using the data structures maintained in the stream, one can answer distance queries at any point in the stream. For any query curve $q \in \XX^d_k$, the algorithm outputs a value $d'(q)$ in query time $\mathcal{O}( k^2 \cdot s \cdot m^{1/s} \cdot f_q(k,\eps',\alpha))$ and additional memory in $\mathcal{O}(k \cdot s \cdot m^{1/s} + k^2)$ such that 
\[\df(q,x) \le d'(q) \le (1+\epsilon)\cdot \df(q,X).\]
\end{theorem}
\begin{proof}
The algorithm reads the input in blocks using calls to the procedure \texttt{process}$(X_i)$, where $X_i$ is the $i$th block.
Let $b$ denote the block size used by the algorithm. We assume $b = k$ and $f(k,\eps) > k$ for any $\eps > 0$. 
We choose $t = \lceil m^{1/s} \rceil$ as the arity of the computation tree.\footnote{Note that this way, choosing $s=\log_2 m$ corresponds to a binary merge tree, since $t = \lceil m^{1/\log_2 m} \rceil = 2$.}  We denote with $\eps'$ the value of $\eps$ which we use to build the data structures at the leaves of the computation tree. 

In addition to a buffer used to store the raw data of the current block, which by our assumption has size $\mathcal{O}(k)$, the algorithm uses additional memory to keep the data structures in $S$ and the simplifications in $A$. By Lemma~\ref{lem:t:merge:depth}, at each point in time, there are at most $(t-1) \cdot \lceil \log_t m\rceil$ data structures stored in $S$ and $A$, each of them taking up space in $f_s(k,\eps')$, respectively $\mathcal{O}(k)$. Moreover, the algorithm uses additional memory when performing a distance query using Lemma~\ref{lem:composition:chain} or Lemma~\ref{lem:composition:build:chain}, which is bounded by $\mathcal{O}(k \cdot (t-1)\log_t m + k^2)$. In addition, each \texttt{merge-reduce} operation uses $\mathcal{O}(t)$ space for handling the data structures to be merged.
Therefore, the additional memory used by the algorithm is bounded by 
\[ 
\mathcal{O}(t \cdot \log_t m \cdot (f(k,\eps',\alpha))).\]
Since $t = \lceil m^{1/s} \rceil$, we have $ s \geq \log_t m $.
Thus, the claimed bound follows.

We now analyze the total running time after reading $m$ points from the input stream. That is, assume the algorithm has read
$\lfloor m/b\rfloor$ blocks from the input and buffered $m \mod b$ points of the current block. 
When handling a block with a call to the \texttt{process} function, the algorithm updates the at most $(t-1) \cdot \lceil \log_t m\rceil$ (Lemma~\ref{lem:t:merge:depth}) $k$-simplifications stored in $A$  using Theorem~\ref{thm:streamingeightapprox} in $\mathcal{O}(k)$ time per simplification per input point. Assuming $b=k$,  the total time spent on updating simplifications is in 
\[\mathcal{O}( m \cdot t \cdot \log_t m ).\]

For each call to the \texttt{process} function, the algorithm uses Theorem~\ref{thm:distoracle} to build a new data structure on the new data $X_i$ in $\mathcal{O}(b \cdot k \cdot f_s(k,\eps',\alpha))$ time.
Since we have $m/b$ blocks, the total time spent on creating data structures at the leaves of the merge tree is in 
\[\mathcal{O}( m \cdot k \cdot f_s(k,\eps',\alpha)).\]

By Lemma~\ref{lem:t:merge:depth}, the algorithm performs in total $\mathcal{O}(m/b)$ calls to the \texttt{merge-reduce} operation, which by Lemma~\ref{lem:composition:build:chain} each take 
$\mathcal{O}(k^2 \cdot t \cdot f_s(k,\eps',\alpha) \cdot f_q(k,\eps',\alpha) )$
time, plus $\mathcal{O}(\log (t \log_t m))$ time for retrieving the $k$-simplification with key value $j_t$ from $A$. Assuming $b=k$, we have in total
\[\mathcal{O}(
m \cdot k \cdot t \cdot f_s(k,\eps',\alpha) \cdot  f_q(k,\eps',\alpha) 
+ (m/b) \log (t \log_t m)
)\]
time spent on the \texttt{merge-reduce} operations. We can see that the time updating simplifications dominates retrieving them and the time spent on creating the data structures at the leaves of the merge tree is dominated by the \texttt{merge-reduce} operations.
The total running time is therefore bounded by 
\[
\mathcal{O}\left( m^{(1+\frac{1}{s})} \left(s +  k \cdot f_s(k,\eps',\alpha) \cdot  f_q(k,\eps',\alpha)\right) \right)
\]

By Lemma~\ref{lem:composition:chain}, the algorithm handles distance queries to $X$ in $\mathcal{O}( k^2 \cdot t \log_t m  \cdot f_q(k,\eps,\alpha))$ time, since $\ell$ is in $\mathcal{O}(t \log_t m)$ by Lemma~\ref{lem:t:merge:depth}. Using  $t = \lceil m^{1/s} \rceil$ and $ s \geq \log_t m $ the claimed bound follows.

In order to guarantee an approximation factor of at most $(1+\eps)$ for the result of the distance queries, we choose $\eps'$ as follows. Note that for each call to the \texttt{merge-reduce} operation, the approximation factor worsens. By Lemma~\ref{lem:t:merge:depth} any data structure in $S$ is the result of at most $s$ basic \texttt{merge-reduce} operations. After at most $s$ \texttt{merge-reduce} operations using Lemma~\ref{lem:composition:build:chain}, the approximation factor has increased to at most $(1+\eps)^{s}$. 
We choose 
\[\epsilon' := \epsilon / (2 s),\] and obtain
\[
\left(1+\frac{\epsilon}{2s }\right)^{s} \le e^{\epsilon/2} \le \frac{1}{1-\epsilon/2} = 1 + \frac{\epsilon/2}{1-\epsilon/2} \le 1 + \epsilon
\]
as an upper bound to the approximation factor of any of the data structures generated by the streaming algorithm while  reading the first $m$ points in the stream.

\end{proof}

\subsection{Results}
We distinguish two main scenarios:  (i) maintaining a dynamic distance oracle in the stream (ii) computing a static data structure in one pass over the data using sublinear additional workspace.

\subsubsection{Streaming algorithm}\label{sec:streaming:result}

In this section we give the result for the scenario of dynamically maintaining a distance oracle in a stream. In this case, we do not know $m$ in advance, and we need to be able to perform distance queries at any point throughout the stream on the part of the sequence seen so far. Since consolidating the data structures in the stack into one data structure is costly, we answer streaming queries using Lemma~\ref{lem:composition:chain}.

\begin{restatable}{theorem}{thmDistOracleStreamingKnownM}
\label{thm:streamingDistOracle:1a}
There is a streaming algorithm that computes a distance oracle for a curve $X$ given as a stream with $m=|X|$, where $m$ is known in advance, with the following properties. The algorithm uses additional memory in 
\[ 
\mathcal{O}\left(\log m \cdot k^k \cdot \left(\log \left(\frac{\log m}{\eps}\right)\right)^k  \cdot \left(\frac{\log m}{\eps}\right)^{dk}\right)
\] and the total time used by the algorithm to compute a distance oracle in the stream is bounded by
\[
\mathcal{O}\left( 
m  \cdot 
k^{k+3}  \cdot \left(\frac{\log m}{\eps}\right)^{dk} \cdot \left(\log \left(\frac{\log m}{\eps}\right)\right)^k
 \right)
\]
Using the data structures maintained in the stream, one can answer distance queries at any point in the stream. For any query curve $q \in \XX^d_k$, the algorithm outputs a value $d'(q)$ in query time $\mathcal{O}( k^4 \cdot \log m + k^2\cdot (\log m) \cdot (\log \epsilon^{-1}))$ and additional memory in $\mathcal{O}(k \cdot \log m + k^2)$ such that 
\[\df(q,x) \le d'(q) \le (1+\epsilon)\cdot \df(q,X).\]
\end{restatable}

\begin{proof}
We use Theorem~\ref{thm:streamingDistOracleTradeoff} with $s=\log m$ and with the data structure of Theorem~\ref{thm:distoracle}, which guarantees $f_s(k,\eps,\alpha)=k^k (\log \eps^{-1})^k  \eps^{-dk}$ and $f_q(k,\eps,\alpha)=k^2 + \log \epsilon^{-1}$. (Note that in this case, $m^{1/s} = (2^{(\log_2 m)}) ^{1/ (\log_2 m)} = 2$.)
\end{proof}

Setting $\epsilon'$ correctly requires knowledge of $m$, the length of the stream. If the length of the stream is unknown, one proceeds like this: Starting with a small estimate $m_0$ on the stream length, one executes the merge-and-reduce framework on the first $m_0$ input points. The result of this is stored away, and the estimate on the stream length is doubled, the next $2 m_0$ items are processed, and so on. In this way, one gets a collection of data structures, one for each guess of the stream lengths. Since the guesses are doubled, the total number is in $\mathcal{O}(\log m)$, so this adds one multiplicative factor of $\log m$ to the asymptotic memory requirement. We can answer distance queries on the stored data structures using Lemma~\ref{lem:composition:chain}. We obtain the following theorem.

\thmDistOracleStreamingUnKnownM*

\subsubsection{Implications}\label{sec:streaming:implications}
We give a general version of our space-time tradeoff result in  Theorem~\ref{thm:streamingDistOracleTradeoff} in Section~\ref{sec:streaming:analysis}. To highlight its implications, we state two applications of it here to curves in the Euclidean space. The first is for the one-pass scenario, where we compute a static data structure, but in a streaming fashion. In this setting, we focus on speeding up short query curves after processing the whole stream. 

\thmOnePassDistOracle*

\begin{proof}
We use Theorem~\ref{thm:streamingDistOracleTradeoff} with parameter $s$ and using the data structure of Theorem~\ref{thm:distoracle}, which guarantees $f_s(k,\eps,\alpha)=k^k (\log \eps^{-1})^k  \eps^{-dk}$ and $f_q(k,\eps,\alpha)=k^2+\log \epsilon^{-1}$.
We have that $f_s(k,\frac{\eps}{s},\alpha)$ is in $\mathcal{O}(k^k (\log \eps^{-1})^k  \eps^{-dk})$, since $s$ is constant.
By Theorem~\ref{thm:streamingDistOracleTradeoff}, this implies a streaming algorithm for computing a distance oracle with the following properties. The algorithm uses additional memory in 
\[ \mathcal{O}( m^{1/s} \cdot k^k (\log \eps^{-1})^k  \eps^{-dk})
\]
and the total time used by the algorithm to compute a distance oracle in the stream is bounded by 
\[
\mathcal{O}\left( 
m^{(1+\frac{1}{s})} \cdot k^{(k+3)} \cdot (\log \eps^{-1})^k \cdot \eps^{-dk}
 \right)
\]

After reading the last point from the stream, we apply Lemma~\ref{lem:composition:build:chain} on the data structures in the tree, using with the simplification with key equal to $1$. Let $\ell$ be the number of data structures after the algorithm has completed processing the stream. By Lemma~\ref{lem:t:merge:depth}, $\ell$ is in $\mathcal{O}(s \cdot m^{1/s})$. By Lemma~\ref{lem:composition:build:chain} this final step takes time in \[ \mathcal{O} \left(f_s\left(k,\frac{\eps}{s}\right)  k^2 \ell k^2\right) \in \mathcal{O}( m^{1/s} \cdot k^{k+4} \cdot (\log \eps^{-1})^k  \cdot \eps^{-dk}) )\] for constant $s$ and takes additional memory in $\mathcal{O}(k \ell)$.
Clearly, this does not affect the asymptotic running time.

The resulting data structure is static and has size $\mathcal{O}(k^k (\log \eps^{-1})^k  \eps^{-dk})$ and can answer queries in  query time in $\mathcal{O}(k^2)$. 
\end{proof}

The second application is a streaming algorithm where the number of points is a priori unknown and the focus is on being able to answer short curve queries at any point during the data stream. The memory requirement of this algorithm is polylogarithmic in $m$. In this case, we do not know $m$ in advance, and we need to be able to perform distance queries at any point throughout the stream on the part of the sequence seen so far. Since consolidating the data structures in the stack into one data structure is costly, we answer streaming queries using Lemma~\ref{lem:composition:chain}.

\begin{theorem}
\label{thm:streamingDistOracle:1b}
There is a streaming algorithm that computes a distance oracle for a curve $X$ given as a stream with $m=|X|$, where $m$ is known in advance, with the following properties. The algorithm uses additional memory in 
\[ 
\mathcal{O}\left(\log m \cdot k^k \cdot \left(\log \left(\frac{\log m}{\eps}\right)\right)^k  \cdot \left(\frac{\log m}{\eps}\right)^{dk}\right)
\] and the total time used by the algorithm to compute a distance oracle in the stream is bounded by
\[
\mathcal{O}\left( 
m  \cdot 
k^{k+3}  \cdot \left(\frac{\log m}{\eps}\right)^{dk} \cdot \left(\log \left(\frac{\log m}{\eps}\right)\right)^k
 \right)
\]
Using the data structures maintained in the stream, one can answer distance queries at any point in the stream. For any query curve $q \in \XX^d_k$, the algorithm outputs a value $d'(q)$ in query time $\mathcal{O}( k^4 \cdot \log m + k^2\cdot (\log m) \cdot (\log \epsilon^{-1}))$ and additional memory in $\mathcal{O}(k \cdot \log m + k^2)$ such that 
\[\df(q,x) \le d'(q) \le (1+\epsilon)\cdot \df(q,X).\]
\end{theorem}

\begin{proof}
We use Theorem~\ref{thm:streamingDistOracleTradeoff} with $s=\log m$ and with the data structure of Theorem~\ref{thm:distoracle}, which guarantees $f_s(k,\eps,\alpha)=k^k (\log \eps^{-1})^k  \eps^{-dk}$ and $f_q(k,\eps,\alpha)=k^2 + \log \epsilon^{-1}$. (Note that in this case, $m^{1/s} = (2^{(\log_2 m)}) ^{1/ (\log_2 m)} = 2$.)
\end{proof}

\section{Approximate Near Neighbors for short query curves}
\label{section:ann}
In this section, we present efficient data structures for the ANN problem, for polygonal curves.

\begin{definition}[ANN problem]\label{Dgenann}
Input are $n$ polygonal curves $P\subset\MMM^m$. Given a distance function $d(\cdot,\cdot)$, $r>0,$ $\eps>0$, preprocess $P$ into a data structure such that for any query polygonal curve $q\in \MMM^k$, the data structure reports as follows:
\begin{itemize}
    \item if there exists a $p\in P$ s.t.~$\d(p,q)\leq r$, then return $p'\in P$ s.t.~$\d(p,q)\leq (1+\eps)r$,
    \item if $\forall p \in P$, $\d(p,q)\geq (1+\eps)r$ then return "no",
    \item otherwise, the data structure either replies with a curve $p\in P$ s.t.~$\d(p,q)\leq (1+\eps)r$, or with "no". 

\end{itemize}
\end{definition}
In particular, the results of this section concern the ANN problem for polygonal curves under the discrete Fr\'{e}chet distance $\d_{dF}$. 
For any polygonal curve $p$, $V(p)$ denotes the set of its vertices. 

\subsection{ANN for short query curves in  Euclidean spaces}\label{subsection:ann:frechet}

In this subsection, we focus in the common case where the underlying metric is the Euclidean metric. We further assume that $r=1$ since we can uniformly scale the ambient space.

Randomly shifted grids constitute the main ingredient of our algorithm. It has been previously observed \cite{Driemel-lshc-17} that randomly shifted grids induce a good partition of the space of curves: with good probability, near curves pass through the same sequence of cells. 
Let $\delta>0$ and $z$ chosen uniformly at random from the interval $[0,\delta]$. The function 
$h_{\delta,z}(x_i)=\left\lfloor {\delta}^{-1}({x_i-z}) \right\rfloor$
induces a random partition of the line. Hence, for any vector $x=(x_1,\ldots,x_d)$, the function 
$g_{\delta,z}(x)=(h_{\delta,z}(x_1),...,h_{\delta,z}(x_d)),$
induces a randomly shifted grid. Notice that, for our purposes, it suffices to use the same random variable for all coordinates. 

For any set $X$, $diam(X)$ denotes the diameter of $X$. We begin with simple technical lemmas and then we proceed to our main theorems. First we  bound the probability that a set with bounded diameter is entirely contained in a cell. 
\begin{lemma}
\label{lemma:probset}
Let $X\subseteq \RR^d$ be a set such that $diam(X)\leq \Delta$. Then,
\[ 
 \Pr_z\left[ \exists x\in X~\exists y \in X:~g_{\delta,z}(x) \neq g_{\delta,z}(y) \right] \leq \frac{d\Delta}{\delta}.
\]
\end{lemma}
\begin{proof}
 Let $a,b\in \RR$ such that $|a-b|\leq \Delta$. Then, 
 \[
 \Pr_{z} \left[\left\lfloor \frac{a-z}{\delta} \right\rfloor \neq \left\lfloor \frac{b-z}{\delta} \right\rfloor  \right] \leq \frac{\Delta}{\delta}.
\]
Hence, by a union bound over all coordinates: 
\[
 \Pr_z\left[ \exists x\in X~\exists y \in X:~g_{\delta,z}(x) \neq g_{\delta,z}(y) \right] \leq \frac{d \Delta}{\delta}.
\]
\end{proof}
The same argument extends to $k$ sets of bounded diameter. 
\begin{lemma}\label{lemma:probksets}
Let $X_1,\ldots,X_k\subseteq \RR^d$ be $k$ sets such that $\forall i\in [k]:~diam(X_i)\leq \Delta$. 
 \[ 
 \Pr_z\left[\exists X_i~ \exists x\in X_i~\exists y \in X_i:~g_{\delta,z}(x) \neq g_{\delta,z}(y) \right] \leq \frac{dk\Delta}{\delta}.
\]
\end{lemma}
\begin{proof}
The statement holds by Lemma \ref{lemma:probset} and a union bound over all sets. 
\end{proof}

The following lemma indicates that the optimal traversal between two polygonal curves $p\in \XX_m^d$ and $q\in \XX_k^d$, $k\leq m$, can be viewed as a matching between $V(p)$ and $V(q)$ with at most $k$ disconnected components.

\begin{lemma}[Lemma 3 \cite{Driemel-lshc-17}]
\label{lemma:kcomponents}
 For any two curves $p\in \XX_{m_1}^d$ and $q \in \XX_{m_2}^d$, there always exists
an optimal traversal $T$ with the following two properties:
\begin{itemize}
 \item[(i)] $T$ consists of at most $k = \min\{m_1, m_2\}$ disconnected components.
 \item[(ii)] Each component is a star, i.e., all edges of this component share a common vertex.
\end{itemize}

\end{lemma}

This allows us to bound the probability of two polygonal curves hitting the same sequence of cells.  

\begin{lemma}
\label{lemma:probfail}
  For any two curves $p \in \XX_m^d$ and $q \in \XX_k^d$, let $X_1^T,\ldots,X_l^T$ be a sequence of subsets of $V(p)\cup V(q)$, where $X_i^T$  denotes the $i$th disconnected component of an optimal traversal $T$. If $\d_{dF}(p,q)\leq 1$, then for $\delta= 4dk$:
  \[
 \Pr_z\left[\exists i \in[d]~ \exists x\in X_i~ \exists y \in X_i:~g_{\delta,z}(x) \neq g_{\delta,z}(y) \right] \leq \frac{1}{2}.
  \]
\end{lemma}

\begin{proof}
 Lemma \ref{lemma:probksets}, and the fact that for any $i\in [k]$ $diam(X_i^T)\leq 2$, imply the result. 
\end{proof}

Hence, by a union bound, we  bound the probability of splitting one of the $k$ disconnected components with a random partition induced by a randomly shifted grid with side-length $\Theta(kd)$. Furthermore, we can precompute and store answers for polygonal curves realized by the grid points of a refined grid of side-length $\Theta(\epsilon/\sqrt{d})$, and use this information to answer any query, after snapping its vertices to the grid.
\thmANNdfdhd*
\begin{proof}
For any vector $x=(x_1,\ldots,x_d)$, we define the random function 
\[g_{\delta,z}(x)=\left(\left\lfloor \frac{x_1-z}{\delta}\right\rfloor,\ldots,\left\lfloor \frac{x_d-z}{\delta}\right\rfloor\right),\]
where $z$ is a random variable following the uniform distribution in $[0,\delta]$, and $\delta=2dk$. 
Let $\mathcal{G}_w^d$ be the canonical uniform grid of sidelength $w$ in $\RR^d$, and 
let  
\[g_{w,\cdot}(x)=\left(w \cdot \left \lfloor \frac{x_1}{w}\right\rfloor,\ldots,w \cdot \left\lfloor \frac{x_d}{w}\right\rfloor\right),\]
be the function which snaps points to $\mathcal{G}_w^d$.
We set 
$w=\eps/ (2\sqrt{d})$.  A natural ordering of the grid points is the lexicographical ordering with respect to their id vectors: for each grid point $(a_1,\ldots,a_d) \in \mathcal{G}_w^d$, its id vector is the vector of integers $(\frac{a_1}{w},\ldots,\frac{a_d}{w})$.

Now, the preprocessing algorithm first samples $z\in[0,\delta]$. Then, the algorithm proceeds as follows:
\begin{itemize}
 \item[(a)] Input: $n$ polygonal curves $P \subset \XX_m^d$.
 \item[(b)] For each curve $p\in P$, assign a key vector $f(p)\in\ZZ^k$ which is defined by the sequence of cells induced by $g_{\delta,z}$, which are stabbed by $p$. The curves which stab more than $k$ cells are not stored.  If the number of stabbed cells is less than $k$, then for the last coordinates we use a special character indicating emptiness.  
 \item[(c)] Store curves in a hashtable: each bucket corresponds to a key vector (as described in (b)).
 \item[(d)] 
 For each key vector in $f(P)$ which corresponds to some sequence of cells $C_1,\ldots,C_t$:  
 \begin{itemize}
     \item for each query representative of complexity $k$ which is defined by points in 
 $ C_1\cap \mathcal{G}_w^d,\ldots,C_t \cap \mathcal{G}_w^d $ (and respects the ordering of cells $C_i$): find the index of some near neighbor within distance $1+\epsilon/2$ or determine that there is none,
     \item store the {answers} (either an index or a special character indicating that there is no near neighbor) in a new hashtable: one new hashtable per bucket of step (c).
 \end{itemize} 
 \end{itemize} 

The query algorithm:
\begin{itemize}
\item[(i)] Input: query curve $q\in \XX_k^d$.
 \item[(ii)] Hash the curve twice: first by $g_{\delta,z}(\cdot)$, and then by $g_{w,\cdot}(\cdot)$. Report the precomputed answer.  
\end{itemize}

\textit{Storage.}
We use perfect hashing to store the curves. There are at most $n$ non-empty buckets which contain curves. For each such bucket, we precompute and store approximate answers for 
all possible queries. The number of query representatives which are compatible with a given sequence of $k$ cells is upper bounded by the number of $k$-choices with repetition from a set of $k\cdot \left( 4d^{3/2}k \eps^{-1}\right)^d$ grid points. This implies a bound of 
$\left(k\cdot \left( 4d^{3/2}k \eps^{-1}\right)^d \right)^k$, which can be slightly improved if we take into account the ordering constraint: 

\[ \sum_{\substack{t_1+\ldots +t_k=k \\ \forall i:~ t_i\geq 0 \\t_1\geq 1,t_k\geq 1 }} \prod_{i=1}^{k} \left(\frac{4d^{3/2}k}{\eps}\right)^{t_i d} \leq 
\sum_{\substack{t_1+\ldots +t_k=k \\ \forall i:~ t_i\geq 0}}
\left(\frac{4d^{3/2}k}{\eps}\right)^{kd}={{2k-1}\choose{k}}\cdot \left(\frac{4d^{3/2}k}{\eps}\right)^{kd} \leq \left(\frac{16d^{3/2}k}{\eps}\right)^{kd}.
\]

Hence there are $n \cdot \mathcal{O}(d^{3/2}k\eps^{-1})^{kd}$ indices to store. Indices refer to the input set of polygonal curves which are stored in $\mathcal{O}(dnm)$. 

\textit{Preprocessing time.} For each data curve, we compute the real distance to all query representatives. Hence, the total preprocessing time is $dnmk\cdot \mathcal{O}\left(\frac{kd^{3/2}}{\eps}\right)^{kd}$. 

\textit{Query time.}
$\mathcal{O}(kd)$ because of perfect hashing.

\textit{Correctness.}
By Lemma \ref{lemma:probfail}, we have that if $\d_{dF}(p,q)\leq 1$, then $p,q$ lie at the same bucket with probability $\geq 1/2$. Now, let any two points $x,y\in \RR^d$, and let $x'$ be the image of $x$ in $G_{\eps/2\sqrt{d}}$. If $\|x-y\|_2 \leq 1 $, then $\|x'-y \|_2 \leq \|x-x'\|_2+\|x-y\|_2 \leq 1+\eps/2$. Similarly, If $\|x-y\|_2 > 1+\eps $ then $\|x-y\|_2 > 1+\eps/2$.
\end{proof}

One may notice that the above data structure requires limited randomness. In fact, there is only one random variable which is used for the randomly shifted grid. As a consequence, the data structure can be easily derandomized, as it is shown in the following theorem.

\thmANNddfdhd*

\begin{proof}
The data structure is very similar to the one of Theorem \ref{TannDFDhd}. We give full details for the sake of completeness. $
 \mathcal{G}_w^d$ denotes the canonical uniform grid of  sidelength $w$
 and we set $w=\eps/ (2\sqrt{d})$ and $\delta=2dk$. Functions $g_{\delta,z}$ and $g_{w,\cdot}$ are defined as in the proof of Theorem \ref{TannDFDhd}. 

 First we snap all points to a grid with side-length $\eps/\sqrt{d}$. By Observation \ref{obs:b}, we can see that this grid snapping only  introduces an additive error of $\pm\eps$ to the discrete Fr\'echet distance of any two curves. Now we proceed as in the  proof of Theorem \ref{TannDFDhd} and we employ a shifted grid with parameters $w,z$. 
 However, instead of using a random value for $z$, we notice that all coordinates are multiples of $\eps/\sqrt{d}$, and since $z \in [0,\delta]$, there are  at most $ \delta \sqrt{d}/\eps+1 = \mathcal{O}(d^{3/2}k/ \eps) $ critical values for $z$. Hence instead of applying a randomly shifted grid, we build several shifted grids; one for each leading to a unique outcome. 
 
 We now give a description of the overall data structure. 
 
 The preprocessing algorithm:
 \begin{enumerate}
 \item Input: $n$ polygonal curves $P \subset \XX_m^d$. 
 \item Snap points to a grid of side-length $\eps/\sqrt{d}$. Let $P'$ be the new set of points.
 \item For each $z = 0,\frac{\eps}{\sqrt{d}},\frac{2\eps}{\sqrt{d}},\ldots, \lfloor\delta \rfloor$:
\begin{enumerate}
     \item For each curve $p\in P'$, assign a key vector $f(p)\in\ZZ^k$ which is defined by the sequence of cells induced by $g_{\delta,z}$ which are stabbed by $p$. The curves which stab more than $k$ cells are not stored.  If the number of stabbed cells is less than $k$, then for the last coordinates we use a special character indicating emptiness.  
 \item Store key vectors $f(P')$ in lexicographical order.
 \item 
For each key vector in $f(P')$ which corresponds to some sequence of cells $C_1,\ldots,C_t$: 
\begin{itemize}
    \item for each query representative of complexity $k$ which is defined by points in 
 $ C_1\cap \mathcal{G}_w^d,\ldots,C_t \cap \mathcal{G}_w^d $ (and respects the ordering of cells $C_i$): find the index of some near neighbor within distance $1+\epsilon/2$ or determine that there is none,
    \item store the {answers} (either an index or a special character indicating that there is no near neighbor) in lexicographical order with respect to the concatenation of their associated id vectors obtained by $g_{w,\cdot}(\cdot)$. 

\end{itemize}
\end{enumerate}
 \end{enumerate}

The query algorithm:
\begin{enumerate}
\item[(i)] Input: query curve $q\in \XX_k^d$.
 \item[(ii)] Snap points to $\mathcal{G}_w^d$. \item[(iii)] For each $z = 0,\frac{\eps}{\sqrt{d}},\frac{2\eps}{\sqrt{d}},\ldots, \lfloor\delta \rfloor$: 
 \begin{enumerate}
     \item  Locate the corresponding key vector $f(q)$: first  evaluate $g_{w,z}()$ on the vertices of $q$ and then perform a binary search on the lexicographically ordered key vectors $f(P')$. 
     \item Locate the precomputed (approximate) answer: first  evaluate $g_{w,\cdot}()$ on the vertices of $q$ and then perform a binary search on the lexicographically ordered query representatives. 
 \end{enumerate}
\end{enumerate}
 
For the performance of our data structure, we rely on the proof of Theorem \ref{TannDFDhd}. The main two differences to Theorem \ref{TannDFDhd} is the choice of $z$ and the use of binary search for locating the index (instead of hashing). Hence,  space, preprocessing  and query time admit a multiplicative overhead of $\mathcal{O}(d^{3/2}k/ \eps)$, which is the number of all critical values of $z$, and query time is also affected by steps (iii)-(a) which  cost $O(\log n)$ time and step (iii)-(b) which costs $O(kd \log (kd/\eps)) $ by the bound on the number of query representatives in the proof of Theorem~\ref{TannDFDhd}.

The discretization step introduces an additive error of $\pm \eps$. Hence, the overall data structure solves the ANN problem with range parameter $r=1-\eps$ and approximation factor $(1+2\eps)/(1-\eps)$. This easily translates to a solution for the ANN problem with range parameter $r=1$ and approximation factor $1+\eps$: we first need to  scale the set of points by $1/(1-\eps)$ and then we use as approximation parameter a sufficiently small multiple of $\eps$. In particular, we can run the above algorithm with approximation parameter $\eps'=\eps/4$, since $\frac{1+2\eps'}{1-\eps'}=\frac{1+\eps/2}{1-\eps/4} \leq 1+\eps$, which only affects our complexity bounds by constant factors.

\end{proof}

\subsection{ANN for short query curves in doubling spaces}
\label{subsection:anndoubling}

In this subsection, we extend our results to metric spaces with bounded doubling dimension. We consider an arbitrary metric space $(\MMM,\d_{\MMM})$.  
Note that for any finite set $X\subset \MMM$, $\lambda_X \leq \lambda_{\MMM}$. 
We present two data structures for the ANN problem of polygonal curves in arbitrary doubling metric spaces, under the discrete Fr\'{e}chet distance. The dataset consists of curves in $\MMM^m$ and queries belong to $\MMM^k$.  The first data structure achieves $\mathcal{O}(k)$ approximation in the black-box model when the doubling dimension is constant, and the second one achieves $(1+\epsilon)$ approximation in the weakly explicit model.

\subsubsection{Preliminaries on net hierarchies}
We now introduce the main algorithmic tool of this section. Our data structure is based on the notion of net-trees. 
\begin{definition}[Net-tree \cite{HM06}]\label{def:nettrees}
 Let $P\subset \MMM$ be a finite set. A {\em net-tree} of $P$ is a tree $T$ whose set of leaves
is $P$. We denote by $P_v \subseteq  P$ the set of leaves in the subtree rooted at a vertex $v \in T$. Associate with each vertex $v$ a point $rep_v \in P_v$. Internal vertices have at least two children. Each vertex $v$ has a level
$\ell(v) \in \ZZ \cup \{-\infty\}$. The levels satisfy $\ell(v) < \ell(\overline{p}(v))$, where $\overline{p}(v)$ is the parent of $v$ in $T$. The levels of
the leaves are $-\infty$. Let $\tau$ be some large enough constant, say $\tau = 11$.
We require the following properties from $T$:
\begin{itemize}
    \item {\em Covering property:} For every vertex $v\in T$: 
    \[P_v \subset b_{\MMM} \left(rep_v ,\frac{2\tau}{\tau-1} \cdot \tau^{\ell(v)} \right). \]
    \item {\em Packing property:} For every vertex $v \in T$ which is not the root,
    \[ b_{\MMM}\left(rep_v,  \frac{\tau-5}{2(\tau-1)} \cdot \tau^{\ell(\overline{p}(v))-1}\right) \cap P \subset P_v   .\]
    \item {\em Inheritance property:} For every  vertex $u \in T$ which is not a leaf, there exists a child $v \in T$ of $u$ such that $rep_u = rep_v$.
\end{itemize}
\end{definition}

\begin{theorem}[Theorem 3.1 \cite{HM06}]\label{theorem:nettrees}
Given a set $P$ of $n$ points in $\MMM$, one can construct a net-tree for $P$ in  $\lambda_P^{\mathcal{O}(1)} n \log n$ expected time.
\end{theorem}
Enhancing the net-tree so that it supports several auxiliary operations leads to the following theorem. 
\begin{theorem}[Theorem 4.4 \cite{HM06}] \label{theorem:annddimpoints}
 Given a set $P$ of $n$ points  in a metric space $\MMM$, one
can construct a data-structure for answering ANN queries (where the quality parameter $\eps$ is provided together with the query). The query time is $\lambda_P^{\mathcal{O}(1)} \log n + \eps^{-\mathcal{O}(\log \lambda_P)}$, the expected preprocessing time is $\lambda_P^{\mathcal{O}(1)}n \log n$, and the space used is $\lambda_P^{\mathcal{O}(1)}n$.
\end{theorem}

\begin{definition}[Pruned net-tree]
Given some pruning parameter $w>0$, we define the pruned net-tree to be a net-tree as in Definition \ref{def:nettrees} which is pruned as follows: for any $v\in T$ such that $P_v \subset b_{\MMM}\left( rep_v, w \right)$, we delete all points in $P_v$, except for $rep_v$ which remains as the single leaf of $v$. 
\end{definition}
We present a data structure for the range search problem on nets, which is entirely based on \cite{HM06}. We note that in order to keep the presentation simple, we make use of the main results there in a black-box manner, but a more straightforward solution is likely attainable.

\begin{theorem}
\label{theorem:doublingrange}

Let $X\subset \MMM$, where $(\MMM,\d_{\MMM})$ is a metric space, and $X$ is the set of $n$ leaves in a pruned net-tree $T$ with pruning parameter $w$ (i.e.\ $X$ is a $\Omega(w)$-net). 
There exists a data structure with input $X$ which supports the following type of range queries:
\begin{itemize}
    \item given $q\in \MMM$, $r>0$, report $b_{\MMM}(q,r)\cap X $.
\end{itemize}
The expected preprocessing time is $\lambda_X^{\mathcal{O}(1)} n \log n$, the space consumption is $\lambda_X^{\mathcal{O}(1)} n$ and the query time is $\lambda_X^{\mathcal{O}(1)}\log n +\lambda_X^{\mathcal{O}(\log (r/w))}$.
\end{theorem}
\begin{proof}

We build a data structure as in Theorem \ref{theorem:annddimpoints}, and we are able to find a $2$-approximate nearest neighbor of $q$ in time $\lambda_X^{\mathcal{O}(1)}\log n$, with expected preprocessing time in $ \lambda_X^{\mathcal{O}(1)} n \log n$ and space in $\lambda_X^{\mathcal{O}(1)} n$. This point is denoted by $q'$. By the triangular inequality, it suffices to seek for the points of 
$b_{\MMM}(q,r)\cap X$ in $b_{\MMM}(q',3r)\cap X$. 

In order to perform a range query for a leaf $q'$, we invoke an auxiliary data structure from \cite{HM06} (see Section 3.5), which, for any query node $v$, allows us to find all 
points $U$ within radius $r'=\mathcal{O}(\tau^{\ell(v)})$ that are roughly at the same level, i.e\ $\forall u \in U:~\ell(u) \leq \ell(v) < \ell(\overline{p}(u)) $. This can be done by maintaining appropriate lists of size $\lambda_X^{\mathcal{O}(1)}$, while building the net-tree, and it does not affect asymptotically the construction of the net-tree.  By the packing property of pruned net-trees, we can retrieve all leaves within distance $\mathcal{O}(r)$ from $q'$ in time $\lambda_X^{\mathcal{O}(\log (r/w))}$. 
\end{proof}

\subsubsection{Preliminaries on random partitions}
\def\PPP{{\mathcal P}}
Given a finite metric space $(X,\d_{\MMM})$, a partition of $X$ is a set $\PPP$ of disjoint subsets of $X$, such that $\bigcup_{Y \in \PPP} Y = X$. We refer to these subsets as clusters. 

Our data structure is based on a  random partition method which is quite common in the literature, especially in  metric embeddings. 
We use this method in order to obtain a partition of the curves with the desired property that near curves probably belong to the same cluster. 
For any set $X$, $diam(X)$ denotes the diameter of $X$.

\def\III{{\mathcal I}}
\algorithmus{\texttt{partition}($X\subset \MMM$, $\Delta >0$)}{
\begin{itemize}
\setlength{\itemsep}{0.0cm} 
\item Set random permutation of $X$: $x_1,x_2, \ldots,x_n$.
  \item Set $C_0 \gets\emptyset$.
  \item Set ordered set $\PPP \gets  \emptyset$.
  \item Choose uniformly at random $R\in [\Delta/4,\Delta/2]$.
  \item {\bf For} $i=1 ,\ldots, n$:
  \begin{itemize}
        \item Set $C_i\gets \{p\in X \mid \d_{\MMM}(x_i,p)\leq R\} \cup C_{i-1}$, where $C_{i-1}\subseteq X$ is the set of covered points in the $(i-1)$th iteration. 
       \item Set $P_i\gets C_i \setminus C_{i-1}$.
          $\PPP \gets  \PPP \cup \{P_i\}$.
  \end{itemize}
  \item {\bf Return} the permutation  $x_1,x_2, \ldots,x_n$, and indices to corresponding clusters according to $\PPP$.\end{itemize}}
The following lemma describes the performance of the above partition scheme.   In other words, we consider cells: each cell is centered at some point $x_j$ and it is defined as the set  $b_{\MMM}(x_j,R) \setminus \left( \bigcup_{i<j} b_{\MMM}(x_i,R) \right)$.   
Typically, similar guarantees discussed in the literature concern only points participating in the procedure (e.g.\ Lemma 26.7 \cite{H11}), while we need to take into account a query point which is not known in advance. 
To that end, we include a proof for completeness. 

\begin{lemma}\label{lemma:metricpartition}
Let $(\MMM,\d_{\MMM})$ be a metric space, $X\subset \MMM$ a finite subset, and let $\PPP$ be the random partition generated by \texttt{partition}$(X,\Delta)$. For any $x\in X$, let $\PPP(x)$ be the cluster to which $x$ has been assigned. Then, the following statements hold:
\begin{itemize}
    \item For any $P\in \PPP$, $diam(P)\leq \Delta$.
    \item Let $q \in \MMM$ and let $x_j\in X$ be such that $j=\min\{i  \mid \d_{\MMM}(q,x_i)\leq R\}$. Then, if $b_{\MMM} ( q,t) \cap X  \neq \emptyset$ and $t\leq \Delta/8$,
    \[
     \Pr[b_{\MMM} ( q,t) \cap X  \not\subseteq \PPP(x_j)]\leq \frac{8t}{\Delta} \ln  \left(|b_{\MMM} ( q,\Delta)\cap X|\right).
    \]
\end{itemize}
\end{lemma}
\begin{proof}
Since $R\leq \Delta/2$, obviously $\forall P \in \PPP:~ diam(P)\leq \Delta$. 

Let $N=|b_{\MMM} ( q,\Delta)\cap X|$ and let $p_1,\ldots,p_m$ be the points in $b_{\MMM} ( q,\Delta)\cap X$ which are ordered in increasing distance from $q$. The probability that a certain point $p_i$ serves as the first center for a cluster that intersects (but does not include) $b_{\MMM} ( q,t)$ is upper bounded by the probability that $R\in [d_{\MMM}(p_i,q) -t,d_{\MMM}(p_i,q)+t]$ and $x_i$ appears before $p_1,\ldots,p_{i-1}$ in the permutation, since otherwise one of the previous clusters would have intersected (and possibly covered) $b_{\MMM} ( q,t)$. Formally,
\[
\Pr[\exists x \in X : b_{\MMM} ( q,t)   \cap \PPP(x) \neq \emptyset  \text{ and } b_{\MMM} ( q,t) \cap X  \not\subseteq \PPP(x)]\leq 
\]
\[ \leq 
\sum_{i=1}^{m} \Pr[R \in d_{\MMM}(p_i,q) \pm t] \cdot \frac{1}{i}\leq  \frac{8t}{\Delta} \ln N.
\]

Finally, since $b_{\MMM} ( q,t) \cap X  \neq \emptyset$ and $t\leq \Delta/8$, 
there exists at least one point which serves as a center for a cluster containing $b_{\MMM} ( q,t)$.

\end{proof}

\begin{lemma}\label{lemma:partition}
Given as input parameters $\Delta>0$, a pruned net-tree $T$ with pruning parameter $w$, where $X$ is the set of $n$ leaves in $T$, \texttt{partition}$(X,\Delta)$ can be implemented to run in 
$\lambda_X^{\mathcal{O}(1)} n\cdot \log n +n \cdot \lambda_X^{\mathcal{O}(\log (\Delta/w))}$ time. 
\end{lemma}
\begin{proof}
By Theorem \ref{theorem:doublingrange}, we can build a data structure which supports range queries: given a point $q\in \MMM$, $R\in [0, \Delta/2]$, we are able to report $\{x\in X \mid \d_{\MMM}(q,x)\leq R\}$ in time $\lambda_X^{\mathcal{O}(1)}\log n +\lambda_X^{\mathcal{O}(\log (R/w))}\leq \lambda_X^{\mathcal{O}(1)}\log n +\lambda_X^{\mathcal{O}(\log (\Delta/w))}$. Hence, for any point $x_i$, we cover and mark points which had not been covered before, and since we need to consider at most $n$ points, the total amount of time needed is $\lambda_X^{\mathcal{O}(1)} n\cdot \log n +n \cdot \lambda_X^{\mathcal{O}(\log (\Delta/w))}$. 

\end{proof}

\subsubsection{Putting everything together}
We extend the results of Section~\ref{subsection:ann:frechet} to handle curves in doubling metrics. We simulate the use of the (coarse) randomly shifted grid, by making use of \texttt{partition}. 
For a partition which is obtained by \texttt{partition} (actually for any partition), each polygonal curve in $\MMM^k$ stabs at most $k$ distinct cells. Using Theorem \ref{theorem:doublingrange}, we are able to build a data structure on the centers of the partition. Then, recovering the cell that some point belongs to, is easy: we perform a $\Delta$-range query for the given point and then we examine all  points inside this range. By the doubling dimension assumption and due to the packing property satisfied by the net points, we know that there are at most  $ \lambda_X^{\mathcal{O}(\log(\Delta/w))}$ points inside this range. One of the points in this range is the center of the cell: we only need to store one index per point which refers to the permutation used by 
\texttt{partition}. If we repeat for all  vertices of the query curve, we get a sequence of at most $k$ cells which corresponds to the unique key vector for this curve. 


\paragraph{The preprocessing algorithm.}
Let $r'$ be the ANN radius search parameter, and let $r:=4r'/3$.  
First, we build a pruned net-tree  on $X:=\bigcup_{p\in P}V(p)$.  Then, we transform it to a pruned net-tree $T$ with pruning parameter $w:=r/4$, 
by visiting at most all nodes and checking which ones should be deleted.  
We build the data structure of Theorem \ref{theorem:doublingrange}
and we run the algorithm of Lemma \ref{lemma:partition} with input $X$, and 
\[{\Delta = 100  \cdot r \cdot (k \log \lambda_X) \log \left( k \log \lambda_X\right)  }.\] 

The output consists of an ordered set of points and the partition. 

We store $P$ in a hashtable as follows. First we compute one vector of indices per curve indicating the corresponding cells.  If one polygonal curve stabs more than $k$ cells, we discard it. If it stabs less than $k$ cells, we use a special character for the remaining coordinates. The polygonal curves are then stored in a hashtable: each bucket is assigned to a key vector of dimension $k$. Any non-empty bucket corresponds to at most $ k$ cells, of diameter $\leq \Delta$.

\paragraph{The query algorithm.}
For any query $q\in \MMM^k$, we perform $k$ $\Delta$-range queries on the leaves of $T$. For each of the $k$ vertices, we explore points within distance $\Delta$, in order to find which point is the first in the permutation used in \texttt{partition}, that covers it. Then we visit the corresponding bucket and we report one curve stored in it.

\begin{theorem}
\label{TannDFDhdDDweak}
Given as input a set of $n$ polygonal curves $P\subset \MMM^m$ in the black-box model, 
the randomized data structure described above solves the $\mathcal{O}(\rho)$-ANN problem under the discrete Fr\'{e}chet distance,  
with {space} in $\lambda_X^{\mathcal{O}(1)} n m$, expected {preprocessing} time in  
$ n\cdot m\cdot \left( \lambda_X^{\mathcal{O}(\log \rho) }+ \lambda_X^{\mathcal{O}(1)} \log (nm) \right) ,$
 and 
query time in 
$k\cdot \left( \lambda_X^{\mathcal{O}(\log \rho) }+ \lambda_X^{\mathcal{O}(1)} \log (nm) \right),$ 
 where $X:= \bigcup_{p\in P}V(p)$, and $\rho:=\rho(\lambda_X,k)\in \mathcal{O}( k\log \lambda_X).$ 
For any query curve $q\in\MMM^k$, the preprocessing algorithm succeeds with constant probability.
\end{theorem}
\begin{proof}
\textit{Preprocessing time.}
A net-tree can be built in expected time $\lambda_X^{\mathcal{O}(1)} nm \log (nm)$ by \cite{HM06}. This complexity also bounds the time we need to visit all nodes. 
Running the algorithm of Lemma \ref{lemma:partition}, costs $\lambda_X^{\mathcal{O}(1)} n \log n+n \cdot \lambda_X^{\mathcal{O}(\log (\Delta/w))}$. 
Finally, by Theorem \ref{theorem:doublingrange}, computing the vector of indices  costs $ m\cdot \left( \lambda_X^{\mathcal{O}(\log (\Delta/w))}+ \lambda_X^{\mathcal{O}(1)} \log (nm) \right)$ time for each curve. Notice that $\mathcal{O}(\log (\Delta/w))= \mathcal{O}(\log (k \log \lambda_X))$.

\textit{Storage.}
We store a net tree, which requires $\lambda_X^{\mathcal{O}(1)}nm$ space, and a hashtable with at most $n$ non-empty buckets containing indices to curves. 

\textit{Query time.} 
Hence we compute the corresponding key vector in time $ k\cdot \left( \lambda_X^{\mathcal{O}(\log (\Delta/w))}+ \lambda_X^{\mathcal{O}(1)} \log (nm) \right)$. We have access to the bucket in $\mathcal{O}(k)$ time. 

\textit{Correctness.}
We claim that the above data structure solves the ANN problem with radius search parameter $3r/4$ and approximation factor $\mathcal{O}(\Delta/r)$. The choice of our pruning parameter implies that if there is a point in the original point set within distance $3r/4$ from some query point, then there is a leaf in the net-tree within distance $r$. 
In order to prove that the approximation factor holds, we make use of Lemma \ref{lemma:kcomponents}, and the fact that the pruning step only induces constant multiplicative error. This implies that if $\d_{dF}(p,q) \leq 3r/4$ then 
there exists an optimal traversal which consists of $k$ components and each component can be covered by a ball of radius $r$ centered at a point of $X\cup V(q)$. By Lemma \ref{lemma:metricpartition}, the probability that \texttt{partition} splits one component is at most
\[
\frac{8r}{\Delta} \ln \lambda_X \cdot \log \frac{8\Delta}{r}
\leq
\frac{8}{100k} \cdot \frac{\log \left(800 (k \log \lambda_X) \cdot \log (k \log \lambda_X)\right)}{\log (k\log \lambda_X)} \leq 
\frac{8}{100k} \cdot \frac{10+2\log \left((k \log \lambda_X) \right)}{\log (k\log \lambda_X)}
\]
$\leq 99/(100k)$,
and by a union bound the probability that $q$ is separated from its  near neighbor is constant.
\end{proof}

\subsubsection{Improving the result for the weakly explicit model}
Now we describe a variant of the above data structure which exploits the power of the weakly explicit model, where we additionally have access to a doubling oracle.

\paragraph{The preprocessing algorithm.}
The first preprocessing step is similar to the one applied in the proof of Theorem \ref{TannDFDhdDDweak}.  
We build a pruned net-tree $T$ on $X:=\bigcup_{p\in P}V(p)$, with pruning parameter $w=\epsilon r$, in expected time $\lambda_X^{\mathcal{O}(1)} nm \log (nm)$.  
We then build the data structure of Theorem \ref{theorem:doublingrange}
and we run the algorithm of Lemma \ref{lemma:partition} with input $X$, and  
\[\Delta= 100  r \cdot (k \log \lambda_X \log ({1}/{\epsilon})) \cdot \log \left( k \log \lambda_X  \log ({1}/{\epsilon}))\right). \]  

We compute one vector of indices per curve indicating the corresponding cells.  If one polygonal curve stabs more than $k$ cells, we discard it. If it stabs less than $k$ cells, we use a special character for the remaining coordinates. The polygonal curves are then stored in a hashtable: each bucket is assigned to a key vector of dimension $k$. Any non-empty bucket corresponds to at most $ k$ cells, of diameter  at most $ \Delta$. For each cell, we use the doubling oracle to find a covering with balls of radius $\eps r$. 
The centers of these balls allow us to consider query representatives so that we can precompute and store answers. 

Consider a non-empty bucket which corresponds to the sequence of cells $C_1,\ldots C_t$, and for each $C_i$ 
let $S_i$ be the set of points returned by the doubling oracle. To construct the set of
query representatives we take all ordered $k$-subsets of  $\bigcup_{i=1}^t S_i$ that respect the ordering of the cells. Each ordered $k$-set corresponds to a query representative. For each query
representative we precompute the answer and store it in a hash table. Here, again, we use the query representative as a key. In particular, we compute the discrete Fr\'echet distance of the query representative to each input curve stored in the bucket and store as answer the index to the input curve which minimizes this distance.

\paragraph{The query algorithm.} 
For any query $q\in \MMM^k$, we perform $k$ $\Delta$-range queries on the leaves of $T$. For any point $x\in V(q)$, we explore points within distance $\Delta$, in order to find which point is the first in the permutation used in \texttt{partition}, which also covers $x$.
This allows us to compute the corresponding key vector and then visit the bucket, where we locate the representative sequence of points..

\begin{theorem}
\label{TannDFDhdDD}
Given as input a set of $n$ polygonal curves $P\subset \MMM^m$ in the weakly explicit model, and an approximation parameter $0<\eps<1/2$, 
the randomized data structure described above solves the $(1+\eps)$-ANN problem under the discrete Fr\'{e}chet distance,  
with {space} in $ \lambda_X^{\mathcal{O}(1)} nm  +  \lambda_{\MMM}^{\mathcal{O}(k \cdot \log \rho)}n$, expected {preprocessing} time in  
$\lambda_X^{\mathcal{O}(1)} nm \log (nm) +  \lambda_{\MMM}^{\mathcal{O}(k \cdot \log  \rho)}\cdot n  m k ,$
 and 
query time in 
$k\cdot \left( \lambda_{\MMM}^{\mathcal{O}(\log \rho)}+ \lambda_X^{\mathcal{O}(1)} \log (nm) \right),$ 
 where $X:= \bigcup_{p\in P}V(p)$, and \[\rho := \rho({\lambda_X,k,\epsilon})
 \in \mathcal{O}\left( \epsilon^{-1} \cdot k \cdot (\log \lambda_X) \cdot \log ({1}/{\eps})\right).\] 
For any query curve $q\in\MMM^k$, the preprocessing algorithm succeeds with constant probability.
\end{theorem}

\begin{proof}

\textit{Preprocessing time.}
Computing the vector of indices for one curve costs \[ m\cdot \left( \lambda_X^{\mathcal{O}(\log(\Delta/w) )}+ \lambda_X^{\mathcal{O}(1)} \log (nm) \right)\] time, which is essentially the time needed to find the corresponding cells. The weakly explicit model  assumes that we are able to access points which $\eps r$-cover a ball of radius $r$  in $\lambda_{\MMM}^{\mathcal{O}(\log (1/\eps))}$ time.  The number of query representatives which are compatible with a given sequence of $k$ cells is at most: 
\[ \sum_{\substack{t_1+\ldots +t_k=k \\ \forall i:~ t_i\geq 0 \\t_1\geq1,t_k \geq 1}} \prod_{i=1}^{k} \lambda_{\MMM}^{t_i\log (\Delta/w)}= 
\sum_{\substack{t_1+\ldots +t_k=k \\ \forall i:~ t_i\geq 0}}
\lambda_{\MMM}^{k\log (\Delta/w)}={{2k-1}\choose{k}}\cdot \lambda_{\MMM}^{k\log (\Delta/w)} \leq \lambda_{\MMM}^{\mathcal{O}(k\log (\Delta/w))}.
\]
For each query representative, we compute its actual distance to the data curves which belong to the same sequence of cells. 
Putting everything together results in preprocessing time in 
$\lambda_X^{\mathcal{O}(1)} nm \log (nm) +  \lambda_{\MMM}^{\mathcal{O}(k \cdot \log  \rho)}\cdot n  m k$, where $\rho\in \mathcal{O}(\eps^{-1} k (\log \lambda_X) \log (1/\eps))$.

\textit{Storage.} We store a net-tree in $\lambda_X^{\mathcal{O}(1)} nm$. We also store a hashtable with at most $n$ non-empty buckets, which correspond to different cells. For each bucket/cell we store a hashtable with $\leq \lambda_{\MMM}^{\mathcal{O}(k\log (\Delta/w))}$ non-empty buckets, one for each query representative. 

\textit{Query time.}
First, we compute the corresponding key vector in time $ k\cdot \left( \lambda_X^{\mathcal{O}(\log(\Delta/w ))}+ \lambda_X^{\mathcal{O}(1)} \log (nm) \right)$. Then, we have access to the bucket in $\mathcal{O}(k)$ time, and we locate the representative sequence of points in $k\cdot\lambda_{\MMM}^{\mathcal{O}(\log (\Delta/w))}$ time.

\textit{Correctness.}We claim that the data structure solves the $(1+\Theta(\eps))$-ANN problem with radius search parameter $(1-2\epsilon)r$. 
In order to prove correctness, we make use of Lemma \ref{lemma:kcomponents} and the fact that approximating the input dataset by the net, only induces $\Theta(\eps r)$ additive error.  
This implies that if $\d_{dF}(p,q) \leq (1-2\eps)r$ then 
there exists an optimal traversal which consists of $k$ components and each component can be covered by a ball of radius $r$ centered at a point of $X\cup V(q)$. The probability {that} \texttt{partition} splits one component is at most 
\[
\frac{8r}{\Delta} \ln \lambda_X \cdot \log \frac{\Delta}{\eps r}
\leq
\frac{8}{100k \log (1/\eps)} \cdot \frac{\log \left(100 \epsilon^{-1}(k \log \lambda_X \log (1/\eps)) \cdot \log (k \log \lambda_X \log (1/\eps))\right)}{\log (k\log \lambda_X \log (1/\eps))}
\]
\[
\leq 
\frac{8}{100k\log (1/\eps)} \cdot \frac{7+\log(1/\eps)+2\log \left((k \log \lambda_X \log (1/\eps)) \right)}{\log (k\log \lambda_X \log (1/\eps))} \leq \frac{9}{10k},
\]
and by a union bound the probability that $q$ is separated from its approximate near neighbor is $\leq1/10$.
\end{proof}

\section{Acknowledgements}
The authors thank Dan Feldman for useful discussions on the topic of this paper. Anne Driemel thanks the Hausdorff Center for Mathematics for their generous support and the Netherlands Organization for Scientific Research (NWO) for support under Veni Grant 10019853. The research of Ioannis Psarros was co-financed by Greece and the European Union (European Social Fund- ESF) through the Operational Programme $\ll$ Human Resources Development, Education and Lifelong Learning $\gg$ in the context of the project ``Strengthening Human Resources Research Potential via Doctorate Research'' (MIS-5000432), implemented by the State Scholarships Foundation (IKY).

 \bibliography{biblio}

\newpage
\appendix
\section{Offline Construction of the Distance Oracle (missing proofs)}\label{appendix:expgrid}

First we show the following lemma for exponential grids for points.

\expgridlemmacombined*

Notice that we provide two different constructions for general metric spaces and the Euclidean space, since we want to be more efficient in the Euclidean case. We start with the latter and show how to build an actual exponential grid. The later construction for the general metric case is similar in spirit, but does not use an explicit grid structure.

\subsection{Exponential grids for points in Euclidean space}\label{sec:euclideanexpgrid}

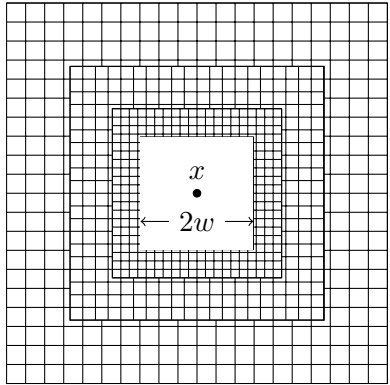
\begin{figure}[bh]
    \centering
\begin{tikzpicture}[scale=0.75]
\coordinate (wmm) at (-1,-1);
\coordinate (wpm) at (+1,-1);
\coordinate (wpp) at (+1,+1);
\coordinate (wmp) at (-1,+1);

\draw[step=0.225*1.5,black,thin] (-2.25*1.5,-2.25*1.5) grid (1.5*2.25,2.25*1.5);
\fill[white] ($1.5*1.5*(wpm)$) rectangle ($1.5*1.5*(wmp)$);
\draw ($1.5*1.5*(wpm)$) rectangle ($1.5*1.5*(wmp)$);
\draw ($2.25*1.5*(wpm)$) rectangle ($2.25*1.5*(wmp)$);

\draw ($2.25*(wpm)$) rectangle ($2.25*(wmp)$);
\draw[step=0.225,black,thin] (-2.25,-2.25) grid (2.25,2.25);
\fill[white] ($1.5*(wpm)$) rectangle ($1.5*(wmp)$);
\draw ($1.5*(wpm)$) rectangle ($1.5*(wmp)$);

\draw[step=0.15,black,thin] (-1.5,-1.5) grid (1.5,1.5);
\draw (wpm) rectangle (wmp);
\fill[white] (wpm) rectangle (wmp);
\node [circle,inner sep=0cm,minimum width=0.1cm,draw, fill=black, label=above:{$x$}] at (0,0) {};
\draw [<->] (-1,-0.5) to node [fill=white] {$2w$} (1,-0.5);

\end{tikzpicture}
    \caption{Exponential grid construction around point $x$ up to the third layer, for $c=1.5$.}
    \label{fig:exponentialgrid}
\end{figure}

In Euclidean space, we denote the ball of radius $r$ around $x \in \mathbb{R}^d$ by $b(x,r)$.

\begin{restatable}{lemma}{expgridlemma}\label{lem:expgrid}
Let $x \in \mathbb{R}^d$ be a point, $\epsilon\in(0,1)$, and $r_1, r_2\in\mathbb{R}$ with $r_1 < r_2$, and assume that $d$ is a constant.
Then there exists a set $\G(x)\subset\mathbb{R}^d$ of size $\mathcal{O}((\ln \frac{r_2}{r_1}) \cdot \epsilon^{-d})$ such that for every $y \in b(x,r_2) \backslash b(x,r_1)$, there is a point $z \in \G(x)$ with \begin{align}
||y-z|| \le \epsilon \cdot ||x-y||.\label{eq:lemma:expgrid}
\end{align}
It is possible to compute the nearest neighbor for a point $y \in b(x,r_2)\backslash b(x,r_1)$ in time $\log |\G(x)|$ (by grid snapping).
\end{restatable}
\begin{proof}
We construct an exponential grid around $x$. To do this, we partition the area between $b(x,r_1)$ and $b(x,r_2)$ into several rectangular shells that we call \emph{layers}. The general construction is: We make a set of nested cubes, define the space between two consecutive cubes as a layer, and then cover each layer by a fine enough grid. We can theoretically ignore all points in $b(x,r_1)$, but since we build a rectangular grid, we cover a little bit of this inner ball by layer one.

Let $R(x,w)$ be the cube with center $x$ and side length $2w$, i.e., 
\[
R(x,w) := \{ y=(y_1,\ldots,y_d)\in\mathbb{R}^d \mid y_i \in [x-w,x+w] \ \forall i \in [d] \}
\]
Fix some constant $c > 1$. We set $w:=r_1/\sqrt{d}$ and define layer $i$ for $i\in[i^\ast]$, $i^\ast:= \lceil \log_c (r_2\cdot \sqrt{d})/r_1\rceil $, by
\[
L_i := R(x,c^{i} w) \backslash  R(x,c^{i-1}w).
\]

We first observe that the layers actually cover $b(x,r_2)\backslash b(x,r_1)$. This is true because the inner area that we do not cover is $R(x,3^0 \cdot w)=R(x,\frac{r_1}{\sqrt{d}}) \subset b(x,r_1)$ and because
\[
3^{i^\ast} w \ge \frac{r_2\cdot \sqrt{d}}{r_1} \cdot w = r_2 
\]
since $i^\ast= \lceil \log_c (r_2\cdot \sqrt{d})/r_1\rceil $.

Next, we cover layer $L_i$ by a grid of width $w_i := \frac{\epsilon}{\sqrt{d}} \cdot c^{i-1} w$ for all $i\in[i^\ast]$. 
More precisely, consider an infinite grid of width $w_i$, and let $\G_i(x)$ be all the grid points of this infinite grid that lie in $L_i$. Then $\G(x) = \cup_{i=1}^{i^\ast} \G_i(x)$.

A sketch of the first three layers around a point $x$ covered by some grid is depicted in Figure~\ref{fig:exponentialgrid}.

The idea of an exponential grid is that the grid becomes coarser at the same rate as the layers grow, i.e., the same number of grip points fall into each layer. This number can be bounded by the number of grid points falling into the outer cube of the layer, which in turn can be bounded by
\[
\left(\left\lfloor \frac{2 \cdot c^i w}{w_i} \right\rfloor +1 \right)^d
\le \left(\frac{2 \cdot c^i w}{\frac{\epsilon}{\sqrt{d}} \cdot c^{i-1} w} +1 \right)^d
= \left( 2 \cdot c \cdot \sqrt{d} \cdot \epsilon^{-1} \right)^d
\]
Thus, the overall number of grid points in all $i^\ast$ layers is bounded by
\[
\lceil \log_c (r_2\cdot \sqrt{d})/r_1\rceil \cdot \left( 2 \cdot c \cdot \sqrt{d} \epsilon^{-1} \right)^d \in \mathcal{O}((\log \frac{r_2}{r_1}) \cdot \epsilon^{-d})
\]
for constant $c$ and $d$.

It remains to show that~\eqref{eq:lemma:expgrid} holds. Let $y\in L_i$ be a point in layer $i$ for $i\in[i^\ast]$. Since $y\in L_i$, its distance to $x$ is at least $c^{i-1}w$. On the other hand, the distance to its closest grid point $z$ is bounded by half the diameter of the grid cells, i.e., by $\sqrt{d}\cdot w_i$. This means that
\[
||y-z|| \le \sqrt{d}\cdot w_i = \epsilon \cdot c^{i-1} w 
\le \epsilon \cdot ||x-y||,
\]
which completes the proof.
\end{proof}

\subsection{Exponential grids for points in doubling spaces}

Let $\MMM$ be a metric space with bounded doubling dimension $d$. 
Assume the weakly explicit model. 
Then we can build a similar structure as in Section~\ref{sec:euclideanexpgrid}. We first observe how many balls we need if we cover with smaller balls of radius $R/X$.

\begin{observation}\label{obs:coveringinmetricspaces}
Any ball of radius $R>0$ in $M$ satisfying the above conditions can be covered by $2 \cdot X^d$ balls of radius $R/X$. 
\begin{proof}
We know that a ball of radius $R > 0$ can be covered by $2^{d}$ balls of radius $R/2$. Each ball of radius $R/2$ can then be covered by $2^{d}$ balls of radius $R/4$, implying that the ball of radius $R$ can be covered by $2^{d} \cdot 2^{d}$ balls of radius $R/4$. For any $i \ge 1$, a ball of radius $R > 0$ can be covered by $2^{i \cdot d}$ balls of radius $R / 2^i$. By setting $i \ge \log X$, we achieve that $R / 2^i \le R/X$. Thus,  $2^{(\log X+1) \cdot d} = 2 \cdot X^d$ balls are sufficient.
\end{proof}
\end{observation}

Now we proceed similarly to Lemma~\ref{lem:expgrid}.

\begin{lemma}\label{lem:expgrid:metric}
Let $(\MMM,\dm)$ be a metric space with constant doubling dimension $d_{\MMM}$. Assume the weakly explicit model. 
Let $x \in \MMM$, $\epsilon \in (0,1)$ and $r_1, r_2 \in \mathbb{R}$ with $r_1 < r_2$. Then  a set $\G(x) \subset M$ of size $\mathcal{O}((\log \frac{r_2}{r_1})\cdot\epsilon^{-d_{\MMM}})$ can be computed in time and space $\mathcal{O}(|\G(x)|)$ such that for every $y \in b_{\MMM}(x,r_2) \backslash b_{\MMM}(x,r_1)$, there is a point $z \in \G(x)$ with
\[
\dm(y,z) \le \epsilon \cdot \dm(x,y).
\]
\end{lemma}
\begin{proof}
 As before, we partition $b_{\MMM}(x,r_2) \backslash b_{\MMM}(x,r_1)$ into shells. Since we can not benefit from grid snapping techniques in the general case anyways, we do now build a rectangular grid, but construct spherical shells. Then each shell is covered by balls of smaller radius, and the centers of these balls form the \lq grid\rq\ structure.
 
 Formally, we cover shell $i$, namely the area $b_{\MMM}(x,2^i r_1) \backslash b_{\MMM}(x,2^{i-1} r_1)$, with balls of radius $r'_i=\epsilon\cdot 2^{i-1} r_1$ for $i\in \{1,\ldots,\bigl\lceil\log \left( \sqrt{r_2/r_2}\right)\bigr\rceil\}$. Since 
 \[
 \epsilon\cdot 2^{\bigl\lceil\log r_2/(\epsilon\cdot r_1)\bigr\rceil}r_1
 \le  \epsilon\cdot r_2/(\epsilon\cdot r_1) r_1 = r_2,
 \]
 it holds that $b_{\MMM}(x,r_2)\backslash b_{\MMM}(x,r_1)$ is completely covered. Covering shell $i$ can be achieved by covering $b_{\MMM}(x,2^i r_q)$ (more points only increase the size of our grid, but will only be beneficial for its quality). By Observation~\ref{obs:coveringinmetricspaces}, we need
 \[
 2 \cdot \left(\frac{2^i r_1} {\epsilon\cdot 2^{i-1} r_1}\right)^{d_{\MMM}}
 = 2 \cdot \left(\frac{2}{\epsilon}\right)^{d_{\MMM}}
 \]
 balls / centers to cover $b_{\MMM}(x,2^i r_2)$ with balls of radius $r_i'$. Thus, the overall number of centers that we get is bounded by
 \[
 \bigl\lceil\log \left( \sqrt{r_2/r_2}\right)\bigr\rceil \cdot 2 \cdot \left(\frac{2}{\epsilon}\right)^{d_{\MMM}}
 \in \mathcal{O}((\log \frac{r_2}{r_2})\cdot\epsilon^{-d_{\MMM}}).
 \]
 It remains to prove that the computed set satisfies the error bound. 
 Let $y \in b_{\MMM}(x,2^i r_1) \backslash b_{\MMM}(x,2^{i-1} r_1)$ be a point. Then $d(x,y) > 2^{i-1}r_1$. Furthermore, the distance to the closest center point $z$ is at most $r_i' = \epsilon \cdot 2^{i-1}r_1$. Thus,
 \[
 \dm(y,z) \le \epsilon \cdot 2^{i-1}r_1 \le \epsilon \cdot \dm(x,y),
 \]
 which is what we wanted to prove.
\end{proof}


\section{A (flawed) approach using simplifications}\label{sec:flawed}

\begin{figure}[bh]
    \centering
    \includegraphics{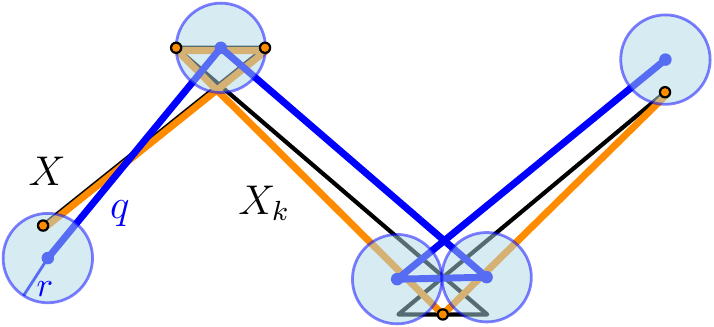}
    \caption{Example of a curve $X$ (black), its simplification $X_k$ (orange) and a query curve $q$ (blue) with distance $r=\df(X,q)$ indicated by the radius of the blue disks. 
}
    \label{fig:nn_grid_frechet}
\end{figure}

In this section we consider building a distance oracle based on a simplification of the input curve. As we discussed in  Section~\ref{sec:prelim}, a simplification $X_k$ of a curve $X$ is a curve with a reduced number of vertices (bounded by a parameter $k$) that has minimal distance to $X$.
Refer to Section~\ref{sec:simplification} for details on how to compute such a simplification.

Now, consider the example given in Figure~\ref{fig:nn_grid_frechet}.
Note that in this example, $C \cdot \df(X,q) < \df(X_k,q)$ for a constant $C > 1$. Therefore, $X_k$ does not contain enough information to provide a $(1+\eps)$-approximation to the distance of $d_{dF}(q,X)$ for arbitrarily small $\eps>0$. 
Intuitively, the input curve here has two \lq zig-zag\rq's, of which a curve with one vertex less can only preserve one. Now the query curve includes a shifted version of the \emph{other} zig-zag. Since the simplification shrinked this to a single point, it can not accurately predict the distance to this shifted zig-zag. 

For any fixed query size $k$, we can extend this example to show that even the optimal $k'$-simplifications of size $k'=m-1$ are not sufficient. This can be described as follows. Let $p_1, p_2, p_3, p_4, p_5, p_6$ be the ordered vertices of curve $X$ in the example. Consider the curve defined by the sequence 
\[ p_1 ~(p_2~ p_3)^t ~p_4~ p_5~ p_6 \]
for some $t \geq 1$.
The sequence has length $m = 4+2t$. Assume that $k=5$, as before, and assume that $\|p_2-p_3\| > \|p_4-p_5\|$. 
Now, if the simplification has significantly more points, say $m > k'> k$, we still obtain a simplification that contracts $p_4$ and $p_5$ to one point, since it is the shortest edge. Using the same arguments as above we can show that it does not contain enough information to provide a $(1+\eps)$-approximation, even if $k' =m-1$.

\section{The doubling dimension of the discrete Fr\'echet distance}\label{sec:doublingspaces}

In order to put our results into context, we also prove a bound on the doubling dimension of the discrete Fr\'echet distance. Note that the doubling dimension of the continuous variant is known to be unbounded~\cite{DriemelKS16}.

\begin{theorem}
\label{theorem:doublDFD}
Let $p=p_1,\ldots,p_m$ be a sequence of points in a metric space $(\mathcal{M},\mathrm{d}_{\mathcal{M}})$ with doubling constant $\lambda_{\mathcal{M}}$. Then, the discrete Fr\'{e}chet ball of radius $r$ and complexity $m$, centered at $p$, $\{q\in (\mathcal{M})^m \mid  \d_{dF}(p,q) \leq r\}$ can be covered by at most $\left(4\cdot \lambda_{\mathcal{M}}\right)^{m}$ discrete Fr\'{e}chet balls of radius $r/2$ and complexity $m$. 
\end{theorem}
\begin{proof}
We consider a set of points $C$ which $r/2$-cover the $m$ Euclidean balls of radius $r$. By the definition of the doubling constant, $\lambda_{\mathcal{M}}$ points per ball  suffice. Notice that points in $C$ are partially ordered, due to the ordering of the $m$ Euclidean balls. 
The discrete Fr\'{e}chet balls of radius $r/2$ and complexity $m$, which are centered at $m$ points in $C$ and satisfy the above ordering constraint, cover $\{q\in (\mathcal{M})^m \mid  \d_{dF}(p,q) \leq r\}$. The number of such discrete Fr\'{e}chet balls is upper bounded by: 
\[
\sum_{\substack{t_1+\ldots +t_m=m \\ \forall i:~ t_i\geq 0 \\t_1 \geq 1, t_m \geq 1}} \prod_{i=1}^{m}{{\lambda_{\mathcal{M}}}\choose{t_i}} \leq 
\sum_{\substack{t_1+\ldots +t_m=m \\ \forall i:~ t_i\geq 0}}
\lambda_{\mathcal{M}}^m={{2m-1}\choose{m}}\cdot \lambda_{\mathcal{M}}^m \leq (4\cdot\lambda_{\mathcal{M}})^m.
\]
\end{proof}

{\em Algorithmic implications.} Previous results for ANN in arbitrary doubling metrics \cite{HM06,CG06} and  
Theorem \ref{theorem:doublDFD} imply data structures for the discrete Fr\'{e}chet distance in arbitrary doubling metrics. 
For datasets consisting of polygonal curves of complexity at most $m$, and query curves of complexity at most $m$, \cite{HM06} implies a data structure with the following performance guarantees: the expected preprocessing time is $\lambda_X^{\mathcal{O}(1)} n \log n $, the space is $\lambda_X^{\mathcal{O}(1)} n$ and the query time is $\lambda_X^{\mathcal{O}(1)} \log n+\left(\frac{1}{\epsilon}\right)^{\log \lambda_X}$, 
where $\lambda_X$ denotes the doubling constant of the input metric. For the $\ell_2^d$ metric, the above result implies a data structure with expected preprocessing time in $2^{\mathcal{O}(dm)} n \log n $, space in $2^{\mathcal{O}(dm)} n$ and the query time in $2^{\mathcal{O}(dm)} \log n+\left(\frac{1}{\epsilon}\right)^{dm}$. We compare these bounds in Table~\ref{tab:compar}. 

{\em Short queries regime.} The above-mentioned data structures fail to provide with improved results when query curves are of much lower complexity than the data curves.
Variants of net-trees like the ones used in \cite{HM06,CG06}, rely on computing good subsets of the dataset with suitable covering and packing properties. The crucial observation there is that it is possible to bound the cardinality of such subsets by $\lambda^{\mathcal{O}(1)}$, where $\lambda$ denotes the doubling constant of the input metric space. Notice that in the discrete Fr\'{e}chet case, these arguments apply to the input dataset of curves of complexity $m$, and they inherently imply complexity bounds depending on $\lambda^{\mathcal{O}(m)}$.

\end{document}